\documentclass[11pt]{article}

\usepackage[letterpaper, margin=1in]{geometry}

\usepackage{setspace}
\usepackage{lmodern}
\usepackage[T1]{fontenc}
\usepackage[utf8]{inputenc}
\usepackage{amsmath, amsfonts, amssymb, amsthm}
\usepackage{algorithm}
\usepackage{algpseudocode}

\usepackage{graphicx}
\usepackage{tikz}
\usepackage{tikz-cd}
\usepackage{multirow}
\usepackage{amscd}
\usepackage{enumerate}
\usepackage{xcolor}
\usepackage{subfig}

\usepackage[colorlinks=true, linkcolor = black, citecolor= black]{hyperref}
\usepackage{pdfcomment}

\theoremstyle{plain}
\newtheorem{theorem}{Theorem}[section]
\newtheorem{nono-theorem}{Theorem}

\newtheorem{lemma}[theorem]{Lemma}
\newtheorem{assertion}{Assertion}

\theoremstyle{definition}
\newtheorem{definition}{Definition}[section]

\newcommand{\X}{\mathbb{X}}
\newcommand{\Y}{\mathbb{Y}}
\newcommand{\R}{\mathbb{R}}

\newcommand{\x}{\mathbf{x}}
\newcommand{\y}{\mathbf{y}}
\newcommand{\Img}{\mathrm{Im}}
\newcommand{\LCA}{\mathrm{LCA}}
\newcommand{\Mf}{\widetilde{M}_f}
\newcommand{\Mg}{\widetilde{M}_g}

\newcommand{\f}{\tilde{f}}
\newcommand{\g}{\tilde{g}}

\title{An Improved FPT Algorithm for Computing the Interleaving Distance between Merge Trees via Path-Preserving Maps}

\author{Althaf P V, Amit Chattopadhyay and  Osamu Saeki}
\date{}

\begin{document}

\maketitle
\tableofcontents
\begin{abstract}
A merge tree is a fundamental topological structure used to capture the sub-level set (and similarly, super-level set) topology in scalar data analysis. The \textbf{interleaving distance}, introduced by \cite{Morozov2013}, is a theoretically sound, stable metric for comparing merge trees.  However, computing this distance exactly is \textbf{NP-hard}. \cite{touli2022fpt} proposed the first \textbf{fixed-parameter tractable (FPT)} algorithm for its exact computation by introducing the concept of an \textbf{$\varepsilon$-good map} between two merge trees, where $\varepsilon$ is a candidate value for the interleaving distance. 
The complexity of their algorithm is $O(2^{2\tau} (2\tau)^{2\tau+2} \cdot  n^2 \log^3 n)$ where \( \tau \) is the degree-bound parameter and $n$ is the total number of nodes in the merge trees.
The complexity is quadratic, with a cubic logarithmic factor of $n$.
Moreover, their algorithm exhibits exponential complexity in  \( \tau \), which increases with the increasing value of $\varepsilon$.  In the current paper, we propose an improved \textbf{FPT algorithm} for computing the $\varepsilon$-good map between two merge trees. Our algorithm introduces two new parameters, $\eta_f$ and $\eta_g$, corresponding to the numbers of leaf nodes in the merge trees $M_f$ and $M_g$, respectively. This parametrization is motivated by the observation that a merge tree can be decomposed into a collection of unique leaf-to-root paths. The proposed algorithm achieves a complexity of
$O\!\left(n^2 \log n + \eta_g^{\eta_f} (\eta_f+\eta_g)\, n \log n \right)$,
which is quadratic with a logarithmic factor in $n$. To obtain this reduced complexity, we assumed $\eta_g^{\eta_f} \leq \eta_f^{\eta_g}$, i.e., the number of possible $\varepsilon$-good maps from $M_f$ to $M_g$ does not exceed that from $M_g$ to $M_f$. Notably, the parameters $\eta_f$ and $\eta_g$ are independent of the choice of $\varepsilon$. Compared to the algorithm proposed in~\cite{touli2022fpt}, our approach substantially reduces the search space for computing an optimal $\varepsilon$-good map and yields a tighter overall complexity bound. We also provide a formal proof of correctness for the proposed algorithm.

\end{abstract}

\section{Introduction}
\label{section:intro}
Merge trees capture the topology of the connected components of the sub-level sets of a scalar field, providing a compact and interpretable topological summary of the scalar function. They are widely used in applications such as symmetry detection~\cite{saikia2014extended,saikia2017fast}, shape retrieval~\cite{sridharamurthy_edit_2023}, feature detection~\cite{saikia_global_2017} and scientific visualization~\cite{yan_structural_2020}. A key challenge in utilizing merge trees for comparative analysis lies in defining and computing meaningful distance measures between them. Several computationally efficient distances have been proposed, including a variant of tree edit distance~\cite{sridharamurthy_edit_2023} and distance based on subtrees of contour trees~\cite{thomas_symmetry_2011}. Merge trees are said to be stable with respect to a distance measure if small perturbations in the scalar function result in small changes in the distance between the Merge trees. While the above measures are attractive for their efficiency, they often lack stability or fail to satisfy the formal properties of a metric.
On the other hand, several theoretically sound stable distance measures have been introduced in the literature, including edit distance~\cite{bauer_edit_2016,bauer_reeb_2021}, interleaving distance~\cite{Morozov2013}, and functional distortion distance~\cite{bauer_measuring_2014}. These metrics are designed to compare both the structural and topological differences between merge trees, and hold considerable promise for applications requiring accurate topological comparison. Even though these metrics are notable for their theoretical soundness, exact computation of these distances is known to be NP-hard~\cite{bollen_reeb_2022}. It has been shown that the interleaving distance is equivalent to the functional distortion distance in the case of merge trees~\cite{bauer_measuring_2014}. The interleaving distance provides a rigorous and stable metric for quantifying structural dissimilarity between merge trees. However, despite its strong theoretical foundation, \cite{agarwal_computing_2018} proved that the interleaving distance is not only hard to compute exactly but also NP-hard to approximate within a factor of 3. In other words, even obtaining a distance estimate that is guaranteed to be at most three times worse than the optimal solution is computationally intractable.  An algorithm is said to be fixed-parameter tractable (FPT) if its runtime becomes efficient when a specific parameter is fixed, even though the problem remains hard in general. In this paper, we address the exact computation of the interleaving distance between two merge trees by fixing appropriate parameters, and we propose a fixed-parameter tractable (FPT) algorithm that improves upon an existing approach.

\textcolor{black}{The definition of the interleaving distance introduced in \cite{Morozov2013} relies on a pair of continuous maps between merge trees, which makes designing an algorithm to compute the interleaving distance between two merge trees challenging.} To address this, \cite{touli2022fpt} proposed a new definition for the interleaving distance using a \emph{single} continuous map, called an \emph{\( \varepsilon \)-good map}, leading to a fixed-parameter tractable (FPT) algorithm for the exact computation of the interleaving distance between two merge trees. The algorithm in \cite{touli2022fpt} chooses the degree-bound (or $\varepsilon$-degree-bound) parameter $\tau$, defined as the maximum total degree of the nodes in any connected portion of the tree lying within a horizontal slab of width $2\varepsilon$, considered over all such slabs from both trees. This parameter $\tau$ is dependent on the candidate value \( \varepsilon \): as \( \varepsilon \) increases, the width of each slab expands, encompassing more of the tree nodes, and consequently increasing the total degree \( \tau \). There exists a set of \( O(n^2) \) number of candidate values for the interleaving distance between two merge trees where $n$ denotes the total number of nodes across both trees~\cite{agarwal_computing_2018}. For any $\varepsilon$ chosen from this set, if there exists an $\varepsilon$-good map, then the interleaving distance is at most $\varepsilon$. While the FPT algorithm in \cite{touli2022fpt} represents a significant theoretical advancement towards computing the exact interleaving distance between two merge trees, its runtime depends exponentially on the degree-bound parameter \( \tau \), which increases with the chosen candidate value \( \varepsilon \). This dependence leads to an unstable complexity bound. In the worst case, when $\varepsilon$ equals the height of the tree, $\tau$ becomes the total degree of all nodes, making the algorithm prohibitively expensive.  
To overcome this limitation, we propose a novel fixed-parameter tractable (FPT) algorithm 
for computing the exact interleaving distance between two merge trees. 
Our approach reparameterizes the problem in terms of the numbers of leaf nodes, 
$\eta_f$ and $\eta_g$, of the merge trees $M_f$ and $M_g$, respectively—parameters that are independent of the choice of $\varepsilon$.
This reparameterization is motivated by a key structural observation: for a continuous and monotone map between two merge trees, the ancestor of a node must be mapped to the ancestor of its image. Since a merge tree admits a unique decomposition into monotone leaf-to-root paths, and the number of such paths equals the number of leaf nodes, this insight naturally suggests a parametrization based on the number of leaf nodes. By enforcing a path-consistency condition, we can then systematically construct an $\varepsilon$-good map between two merge trees.
Reparameterizing the algorithm in terms of   \( \eta_f \text{ and } \eta_g \) which are independent of the choice of $\varepsilon$, yields a stable complexity bound while preserving the correctness and theoretical soundness of the interleaving distance.

 \subsection*{Our Contributions}
We propose a new fixed-parameter tractable (FPT) algorithm for computing the exact interleaving distance between two merge trees. Building on the $\varepsilon$-good map framework introduced in \cite{touli2022fpt}, our algorithm exploits the structural simplicity of merge trees to significantly reduce computational complexity. The core insight driving our improvement is that merge trees admit a unique decomposition into monotone leaf-to-root paths.  This structural property enables us to streamline both the construction of  maps and the verification of whether a given map qualifies as an $\varepsilon$-good map. Our main contributions are summarized as follows.

\begin{itemize}
    \item Theoretically, we establish a monotonicity property of the $\varepsilon$-good map (Lemma~\ref{lemma: monotonicity}), 
    which enables the use of a binary search strategy in our algorithm. 
    We further present a technique for constructing continuous maps between two merge trees based on the gluing lemma~\cite{armstrong1983basic} (Theorem~\ref{theorem:continuity}). 
    Finally, we derive computable criteria for verifying the Ancestor-Shift (Theorem~\ref{thm:ancestor-shift-condition}) and Ancestor-Closeness (Lemma~\ref{lemma:ancestor-closeness}) properties of the constructed maps, 
    which determine whether a given map qualifies as an $\varepsilon$-good map.

    \item  Algorithmically, we propose a novel FPT algorithm for computing the exact interleaving distance between two merge trees by introducing a new parameterization based on the numbers of leaf nodes \( \eta_f \text{ and } \eta_g\) of the merge trees $M_f$ and $M_g$, respectively, which are independent of the choice of \( \varepsilon \), leading to a stable complexity bound. Our algorithm employs a path-based construction: it begins with mappings which are initially defined on the set of leaf nodes of the first merge tree and extends them along the unique leaf-to-root paths to build a map on the entire tree, then it verifies the $\varepsilon$-goodness of the map.


    \item  Our algorithm achieves a total runtime of 
$O(n^2 \log n + \eta_g^{\eta_f} (\eta_f + \eta_g)\, n \log n)$, 
improving upon the previous best bound of 
$O(2^{2\tau} (2\tau)^{2\tau+2} n^2 \log^3 n)$ 
reported in ~\cite{touli2022fpt}. \textcolor{black}{Here, $n$ denotes the total number of nodes across the merge trees $M_f$ and $M_g$, 
which have $\eta_f$ and $\eta_g$ numbers of leaf nodes, respectively, 
and $\tau$ represents the degree-bound parameter used in the previous algorithm. 
The algorithm in ~\cite{touli2022fpt} has quadratic complexity with a cubic logarithmic factor in $n$, 
whereas our proposed algorithm remains quadratic with only a single logarithmic factor. 
Furthermore, the algorithm in~\cite{touli2022fpt} exhibits exponential dependence on the parameter $\tau$, 
which increases with larger values of $\varepsilon$. 
In contrast, the parameters $\eta_f$ and $\eta_g$ used in our algorithm are independent of the choice of $\varepsilon$, 
where $\varepsilon$ denotes a candidate value for the interleaving distance.}


    \item  We provide a formal proof of correctness of the proposed algorithm, ensuring its theoretical soundness. \textcolor{black}{The proof rigorously demonstrates that the algorithm verifies whether two merge trees are $\varepsilon$-interleaved by systematically exploring all possible $\varepsilon$-good maps between them to identify the existence of at least one such mapping. Through a binary search over the finite set of candidate $\varepsilon$ values, the algorithm ultimately terminates with the exact interleaving distance between the merge trees.} 
\end{itemize}

\paragraph*{Overview.}  
Section~\ref{sec:related-work} reviews the relevant literature related to the problem under study. 
Section~\ref{sec:background} presents the foundational concepts necessary for understanding our approach, 
including definitions of merge trees, the interleaving distance, $\varepsilon$-compatible and $\varepsilon$-good maps, 
as well as an outline of the FPT algorithm in \cite{touli2022fpt}. 
Section~\ref{sec:theory} introduces the main theoretical contributions of this work, 
including the monotonicity property of the interleaving distance, the construction of continuous maps between merge trees, 
and computable criteria for the Ancestor-Shift and Ancestor-Closeness properties. 
Section~\ref{Sec:algorithm} details our proposed algorithm for computing the exact interleaving distance between two merge trees. 
Section~\ref{sec:complexity} analyzes the computational complexity of the algorithm by examining each of its components. 
Section~\ref{sec:discussion} compares the performance of our algorithm with existing methods in the literature based on our complexity result. 
Finally, Section~\ref{sec:conclusion} summarizes the main contributions and outlines potential directions for future work.

\section{Related works}
\label{sec:related-work}
Several distance measures have been proposed for comparing merge trees, many of which are computationally efficient and practical for applications \cite{yan_scalar_2021}.  \cite{sridharamurthy_edit_2023} introduced a variant of tree edit distance between merge trees that defines a persistence-aware tree edit framework capturing both structural and scalar field variations.  Although this distance is a valid metric and supports applications such as periodicity detection and symmetry analysis, its stability properties remain underexplored. \cite{bremer_measuring_2014} introduced a branch decomposition distance that identifies the optimal pair of branch decompositions by minimizing the matching cost between them over all possible branch decompositions. \cite{thomas_symmetry_2011} proposed a similarity measure based on topological persistence and branch decomposition overlap that is used to group similar subtrees of a contour tree.  Other approaches include histogram-based method~\cite{saikia2017fast} and  tree alignment method~\cite{lohfink_fuzzy_2020}. While these measures offer computational advantages, they often lack formal guarantees regarding stability or metric validity. Bottleneck distance between the persistence diagrams~\cite{cohen2007stability} of the merge trees offers an efficient and stable alternative, however, the interleaving distance between merge trees offers a more discriminative measure than the bottleneck distance~\cite{bauer_measuring_2014}.

A merge tree with a set of labels on its nodes is termed a labeled merge tree. For labeled merge trees, \cite{gasparovic_intrinsic_2022} proposed an \( O(n^2) \) algorithm to compute the labeled interleaving distance under fixed label constraints, and showed that the (unlabeled) interleaving distance is the infimum over all such labelings. However, enumerating all possible label assignments remains computationally infeasible in general. Reeb graphs are the graph-based topological summaries that capture the evolutions of the level sets of a scalar field, and distances between Reeb graphs can naturally be restricted to merge trees. Edit distance~\cite{di_fabio_edit_2016,bauer_edit_2016,bauer_reeb_2021}, functional distortion distance~\cite{bauer_measuring_2014}, and interleaving distance~\cite{de_silva_categorified_2016} have been proposed between Reeb graphs with strong theoretical foundations.  
However, exact computation of these distances is known to be NP-hard~\cite{bollen_reeb_2022}. Table~\ref{tab:reeb_comparison} summarizes these distance measures and their complexities.

\begin{table}[ht]
\centering
\small 

\begin{tabular}{|l|c|c|c|}
\hline
\textbf{Distance measures} & \textbf{Type} & \textbf{Complexity} \\ \hline

Edit distance~\cite{di_fabio_edit_2016,bauer_reeb_2021} & Metric & NP-hard \\ \hline

Interleaving distance between merge trees\cite{Morozov2013} & Metric & NP-hard \\ \hline
Functional distortion distance~\cite{bauer_measuring_2014} & Metric & NP-hard \\ \hline

 A variant of edit distance between merge trees~\cite{sridharamurthy_edit_2023} & Metric & $O(n^2)$ \\ \hline

Interleaving distance for labeled merge trees~\cite{gasparovic_intrinsic_2022}  & Metric & $O(n^2)$ \\ \hline

Distance based on branch decompositions~\cite{bremer_measuring_2014} & Unknown & $O(n^5 \log(I_\varepsilon))$ \\ \hline
Distance based on histograms for merge trees~\cite{saikia2017fast}  &Metric & $O(n^2B)$ \\ \hline
Distance based on tree alignment~\cite{lohfink_fuzzy_2020}  & Unknown & $O(n^2)$ \\ \hline
Distance based on subtrees of contour trees~\cite{thomas_symmetry_2011}  & Unknown & $O(t^5 \log t)$ \\ \hline

\end{tabular}
\caption{Comparison of distance measures for merge trees and similar topological descriptors. Here, \( n \) denotes the number of nodes in the topological descriptor, \( B \) is the number of bins  in a histogram, 
\( I_\varepsilon \) represents the search range for the optimal pair of branch decompositions, 
and \( t \) is the number of branches in a branch decomposition.
}

\label{tab:reeb_comparison}
\end{table}

\textbf{Approximation Algorithms:} Recently, Pegoraro~\cite{pegoraro_graph-matching_2024} introduced a combinatorial perspective on the interleaving distance by defining a global node matching to approximate the distance between two merge trees. It provides upper and lower bounds on the interleaving distance by leveraging structural alignments. \cite{agarwal_computing_2018} investigated the approximability of both the Gromov--Hausdorff and interleaving distances. The authors provided a polynomial-time algorithm achieving an \(O(\sqrt{n})\)-approximation of the interleaving distance between merge trees. They showed that not only computing the interleaving distance, but also computing an approximation within a factor of 3 is also NP-hard. 

\textbf{FPT Algorithm:} In contrast to these approximation approaches, \cite{touli2022fpt} proposed a fixed-parameter tractable (FPT) algorithm for computing the exact interleaving distance in an attempt to approximate Gromov-Hausdorff distance between metric trees. The authors proposed a new definition of interleaving distance using a single \(\varepsilon\)-good map instead of a pair of maps in the classical definition~\cite{Morozov2013} .
The algorithm presented in their work has a time complexity of $O(2^{2\tau} (2\tau)^{2\tau+2} n^2  \log^3 n)$,
where \(n\) denotes the number of nodes in the merge tree, and \(\tau\) is a degree-bound parameter. However, \(\tau\) increases with the chosen value of \(\varepsilon\), causing the algorithm’s complexity to grow accordingly. In this paper, we improve upon this approach by introducing new parameters, \( \eta_f \text{ and }\eta_g \), which denote the numbers of leaf nodes in the input merge trees $M_f$ and $M_g$ respectively. Notably, \(\eta_f \text{ and }\eta_g \) are independent of \(\varepsilon\), leading to a more efficient and stable complexity bound.

\section{Background}
\label{sec:background}
In this section, we provide the necessary background on merge trees and the computation of the interleaving distance between merge trees, which is essential for understanding our method. 

\subsection{Merge Tree of a Scalar Field}
A \emph{merge tree} $M_f$ associated with a scalar field 
\( f : \X \to \mathbb{R} \) 
on a topological space \(\X\) records how the connected components of the sublevel sets 
\[
  \X_t = f^{-1}((-\infty, t])
\]
appear and merge as the threshold \(t\) increases. 
In the current paper, we assume, the underlying space \(\X\) is \emph{connected and finitely triangulable} (i.e., homeomorphic to a finite simplicial complex) and the scalar field is chosen to be “tame” (for instance, a piecewise-linear function on the $\X$), i.e. the number of distinct critical values of 
$f$ is finite, and new components can only appear or merge at those finitely many values. Between critical values the connectivity remains unchanged, so the evolution of components can be encoded as a finite set of nodes (the critical events) connected by edges (the intervals of persistence). In this setting the merge tree is thus a finite one-dimensional CW complex or, in particular, a finite graph with a tree structure, capturing the hierarchical merging of connected components in a combinatorial form \cite{dey_computational_2022}.

Formally, we define an equivalence relation on $\X$ as  \textcolor{black}{follows}: for $\x,\y \in \X$, we define \( \x \sim \y \) if and only if  \( f(\x) =f(\y) = a \) for some \( a \in \mathbb{R} \), and $\x$ and $\y$ lie in the same connected component of \( \X_a \). This equivalence relation partitions $\X$ into equivalence classes.  The \textit{merge tree} \( M_f \) is then defined as the quotient space \( \mathbb{X}/ \sim \) constructed by collapsing each equivalence class $[\x] = \{\y\in \X : \x \sim \y\}$ into a single point. That is, each point $x$ of $M_f$ corresponds to an equivalence class $[\x]$, and the quotient map $q_f:  \mathbb{X} \rightarrow  M_f $ sends points of each equivalence class $[\x]$ to  $x\in M_f$. 
Thus, the merge tree $M_f$ is a rooted tree in which each node corresponds to a critical point of the function \( f \) where the topology of its sub-level set changes. Since $\X$ is assumed as finitely triangulable, the root of the $M_f$ corresponds to the global maximum of $f$. As the threshold \( t \) increases, a new component appears in $\X_t$ at each local minimum, which corresponds to a degree 1 node in $M_f$ referred as \emph{leaf} node, and two existing components in $\X_t$ merge at a saddle, which corresponds to a node in $M_f$ with degree 3 or more. Thus, a Merge tree provides a compact topological summary that encodes the evolution of connected components of the sub-level sets of a scalar field as the scalar value varies. See an example illustration in the Figure \ref{fig:merge_tree}.

\begin{figure}[htbp]  
    \centering
    \includegraphics[width=0.7\linewidth]{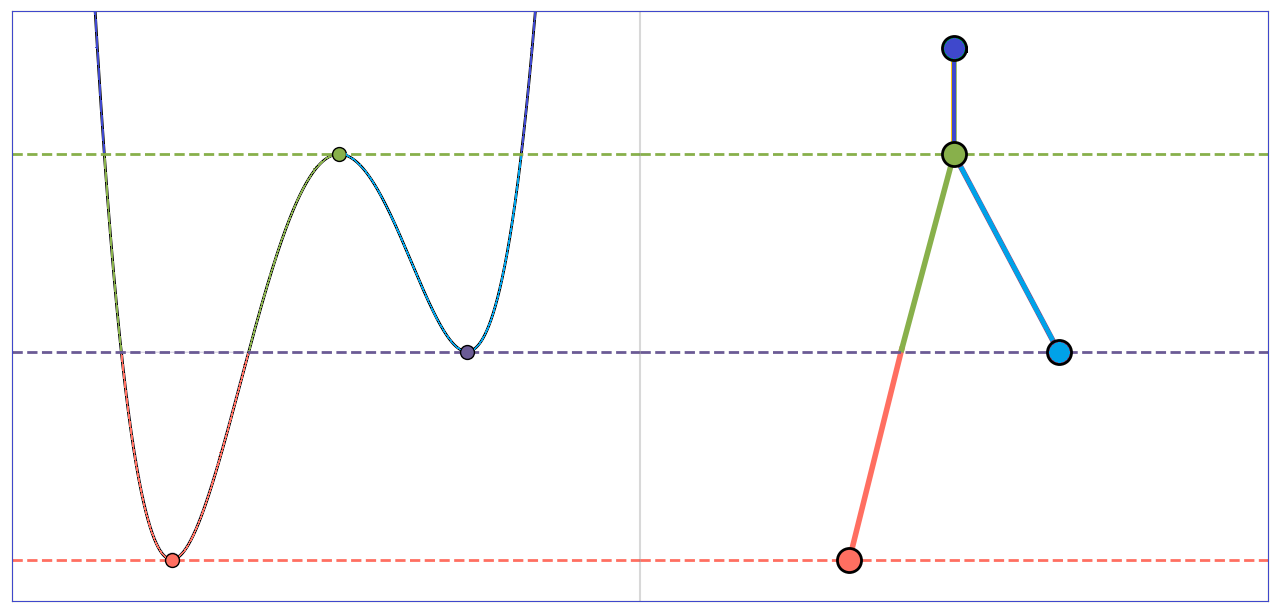}
    \caption{An example of a merge tree of the height field on the graph $(x, f(x))$  of the function $f: \R \rightarrow \R$ defined by $f(x) = x^4 - 4x^2 + x$. The corresponding merge tree is shown on the right, with function values increasing toward \( \infty \) at the root node colored in blue.  The coral (light red) and sky-blue colored nodes in the merge tree correspond to the local minima in the domain, where new connected components are born. The green node in the merge tree corresponds to the saddle where two components merge. The coral edge in the merge tree corresponds to the coral region in the domain, the sky-blue edge to the sky-blue region, the green edge to the green region and the blue edge to the blue region.}
    \label{fig:merge_tree}
\end{figure}

 The function $f:\X \rightarrow \R$  naturally induces a function $\tilde{f}: M_f \rightarrow \mathbb{R}$ defined as follows: for any $x\in M_f$, let $\tilde{f}(x)= f(\y)$ for any point  $\y$ in the equivalence class $[\x]=q_f^{-1}(x)$. 
 This is well-defined because all points in the same equivalence class share the same $f$-value. As a result, the following diagram commutes, that is $f = \tilde{f} \circ q_f$.

\[
\begin{tikzcd}
\mathbb{X} \arrow[r,"f"] \arrow[dr,"q_f"'] & \mathbb{R} \\
& M_f \arrow[u,"\tilde{f}"']
\end{tikzcd}
\]

The function $\tilde{f}$ defined on the merge tree \( M_f \), is used to formulate the interleaving distance between merge trees.

\subsection{Interleaving Distance in \cite{Morozov2013}} 
\label{subsec: Background - Interleaving}
Intuitively, if two merge trees are nearly the same, then corresponding to every point in one tree, it is possible to find a point in the other tree roughly having almost the same function value, such that the overall topological structure of the trees is approximately preserved.
 The idea of \(\varepsilon\)-compatible maps formalizes this: it asks whether we can shift each point in one tree by at most \( \varepsilon \) in function value to find a continuous correspondence in the other tree, and vice versa where $\varepsilon > 0$.  For the notion of $\varepsilon$-compatibility to be well-defined, each merge tree is extended by attaching an infinite ray at the root so that the function value goes to $+\infty$ which is known as the extended merge tree \cite{touli2022fpt}. From now on, by a merge tree we mean an extended merge tree.
The following definition formalizes the \(\varepsilon\)-compatible map.

\begin{definition}
     Let $f : \X \rightarrow \R$ and $g : \Y \rightarrow \R$ be two scalar fields, and let \( M_f \) and \( M_g \) denote the corresponding merge trees. Two continuous maps $\alpha:M_f \rightarrow M_g$ and $\beta:M_g \rightarrow M_f$ are said to be \textbf{$\varepsilon$-compatible} if for all points $x \in M_f$ and $y\in M_g$,
     
\begin{enumerate}[(i)]
    \item $\tilde{g}(\alpha(x)) = \tilde{f}(x)+\varepsilon,$  
    \item $\tilde{f}(\beta(y)) = \tilde{g}(y)+\varepsilon,$
    \item $\beta \circ \alpha = i^{2\varepsilon},$
    \item $\alpha \circ \beta = j^{2\varepsilon},$ 
\end{enumerate}
   where $i^{2\varepsilon}: M_f\rightarrow M_f$ is the $2\varepsilon$-shift map which maps a point $x\in M_f$ to its ancestor $x'\in M_f$, lying on the unique path from $x$ to the root of $M_f$, such that 
   if $\tilde{f}(x)=a$ then $\tilde{f}(x')=a+2\varepsilon$. The map $j^{2\varepsilon}$ is defined analogously on $M_g$.  
   If $\varepsilon$-compatible maps exist between two merge trees, we say those merge trees are \(\varepsilon\)-interleaved. The smaller the required \( \varepsilon \), the more similar the trees are. 

   \label{def:epsilon-comaptible}
\end{definition}

\begin{figure}[h]
    \centering
    \includegraphics[width=0.5\linewidth]{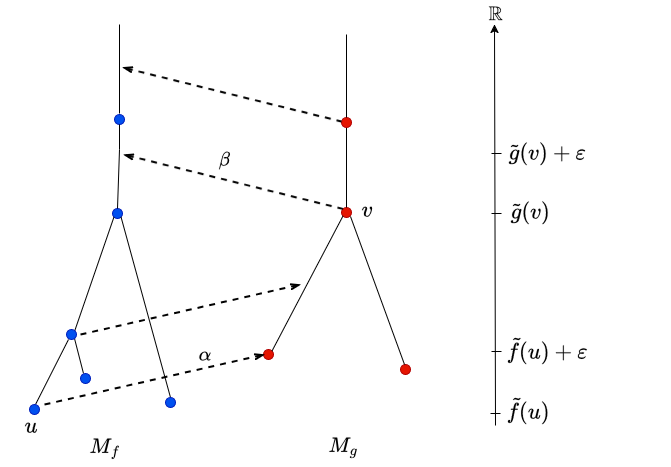}
 \caption{
     Illustration of construction of $\varepsilon$-compatible maps between the (extended) merge trees \( M_f \) and \( M_g \). A blue node $u$ in \( M_f \) is mapped to  a red node $\alpha(u)$ in \( M_g \), satisfying \( \tilde{g}(\alpha(u)) = \tilde{f}(u) + \varepsilon \). Conversely, a red node $v$ in \( M_g \) is mapped to a point $\beta(v)$ in \( M_f \), satisfying \( \tilde{f}(\beta(v)) = \tilde{g}(v) + \varepsilon \). 
}
    \label{fig:interleaving}
\end{figure}

For \( \varepsilon = 0 \), the maps \( \alpha \) and \( \beta \) become mutual inverses, implying that \( M_f \) and \( M_g \) are function-preserving isomorphic. That is, there exists a bijection between the merge trees that preserves both the adjacency structure and the function values.
 If $M_f$ and $M_g$ are \(\varepsilon\)-interleaved, that is there exist $\varepsilon$-compatible maps $\alpha:M_f \rightarrow M_g$ and $\beta : M_g \rightarrow M_f$, then we say $M_f$ and $M_g$ are at most $\varepsilon$ apart from being function-preserving isomorphic. 
 For a clearer understanding of Definition~\ref{def:epsilon-comaptible}, see the illustration in Figure~\ref{fig:interleaving}. 
 The interleaving distance is defined formally as follows.
\begin{definition}
    Interleaving distance between two merge trees $M_f$ and $M_g$, denoted by $d_I(M_f,M_g)$, is defined as
    $$d_I(M_f,M_g)= \inf \left\{\varepsilon \mid \text{there exist $\varepsilon$-compatible maps } \alpha:M_f \rightarrow M_g, ~\beta:M_g \rightarrow M_f\right\}.$$
\end{definition}
We note, $\varepsilon$-compatible maps $\alpha$ and $\beta$ between two merge trees may not always exist for an arbitrary $\varepsilon > 0$. 
However, such maps  always exist for sufficiently large $\varepsilon$, and the infimum of the set of possible $\varepsilon$ values for which they exist is precisely the interleaving distance between the two merge trees.
   The use of a pair of continuous maps in this definition makes it challenging to design an algorithm to compute the interleaving distance between two merge trees. To address this, \cite{touli2022fpt} introduced an equivalent definition based on a single continuous map, referred to as \emph{\( \varepsilon \)-good map}. Next, we present the notion of $\varepsilon$-good map and the fixed-parameter tractable (FPT) algorithm developed in \cite{touli2022fpt} for the exact computation of the interleaving distance between \textcolor{black}{two} merge trees.

\subsection{FPT-Algorithm in \cite{touli2022fpt}}
\label{subsec: FPT algorithm}
\cite{touli2022fpt} proposed a fixed-parameter tractable (FPT) algorithm to compute the exact interleaving distance based on \emph{$\varepsilon$-good map} between the merge trees. Throughout this paper, let $f : \X \rightarrow \R$ and $g : \Y \rightarrow \R$ be two scalar fields defined on finitely triangulable topological spaces $\X$ and $\Y$, respectively. Let \( M_f \) and \( M_g \) denote the corresponding merge trees. For two points $u, v$ in a merge tree, we adopt the notation \( v \succeq u \) to indicate that \( v \) is an ancestor of \( u \) in the merge tree. The formal definition of $\varepsilon$-good map is as follows:

\begin{definition}{}
    For an $\varepsilon>0$, a continuous map $\alpha^{\varepsilon}: M_f \rightarrow M_g$ is called an \textbf{$\varepsilon$-good map} if the following properties hold.
    \begin{enumerate}[(a)]
        \item  \textbf{Range-Shift:} $\tilde{g}(\alpha^{\varepsilon}(v))= \tilde{f}(v)+\varepsilon$, $ \forall v\in M_f$.
        
        \item \textbf{Ancestor-Shift:} If $\alpha^{\varepsilon}(v_1)\succeq \alpha^{\varepsilon}(v_2)$, then $i^{2\varepsilon}(v_1)\succeq i^{2\varepsilon}(v_2)$, $ \forall v_1,v_2\in M_f$.
        
        \item  \textbf{Ancestor-Closeness:} For any $w\in M_g\setminus \Img(\alpha^{\varepsilon})$ and its nearest ancestor $w^a \in \Img(\alpha^{\varepsilon})$, we have $|\tilde{g}(w^a)-\tilde{g}(w)| \leq 2\varepsilon$.
    \end{enumerate}
    \label{def:epsilon-goodmap}
\end{definition}
 Note that the nearest ancestor $w^a$ of a point $w$, required in the \textbf{Ancestor-Closeness} property, always exists because the root of $M_g$ lies in $\mathrm{Im}(\alpha^{\varepsilon})$, and the root necessarily belongs to the ancestor chain of $w$.
   Here, \( i^{2\varepsilon}: M_f \rightarrow M_f \) is the \( 2\varepsilon \)-shift map as defined in Definition~\ref{def:epsilon-comaptible}. 
   Figures~\ref{fig:ancestor-shift} and \ref{fig:ancestor-closeness}  illustrate respectively the Ancestor-Shift and Ancestor-Closeness properties of an $\varepsilon$-good map between two merge trees.
 \cite{touli2022fpt} showed that the existence of a pair of \( \varepsilon \)-compatible maps \( \alpha: M_f \rightarrow M_g \) and \( \beta: M_g \rightarrow M_f \) is equivalent to the existence of a single \( \varepsilon \)-good map \( \alpha^{\varepsilon}: M_f \rightarrow M_g \). From the map $\alpha^\varepsilon$, a corresponding map \( \beta^{\varepsilon}: M_g \rightarrow M_f \) can be constructed such that $\alpha^\varepsilon$ and $\beta^\varepsilon$ are $\varepsilon$-compatible. A map \( \phi : M_f \rightarrow M_g \) is said to be \emph{monotone} if \( \tilde{g}(\phi(x)) \geq \tilde{f}(x) \) for every point \( x \in M_f \), that is, \( \phi \) maps a point in \( M_f \) to a point in \( M_g \) with greater or equal function value. Under this definition, both \( \varepsilon \)-good and \( \varepsilon \)-compatible maps are monotone.

 \begin{figure}[htbp]
    \centering
    \begin{minipage}{0.42\textwidth}
        \centering
        \includegraphics[width=\textwidth]{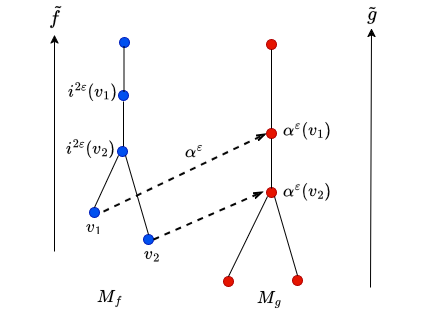}
\caption{
    Illustration of the  Ancestor-Shift property in Definition~\ref{def:epsilon-goodmap} for an $\varepsilon$-good map \( \alpha^\varepsilon: M_f \rightarrow M_g \). Here \( \alpha^\varepsilon(v_1) \succeq \alpha^\varepsilon(v_2) \), then it must hold that \( i^{2\varepsilon}(v_1)\succeq i^{2\varepsilon}(v_2) \) for  \( v_1, v_2 \in M_f \).
}
        \label{fig:ancestor-shift}
    \end{minipage}
    \hfill
    \begin{minipage}{0.42\textwidth}
        \centering
        \includegraphics[width=\textwidth]{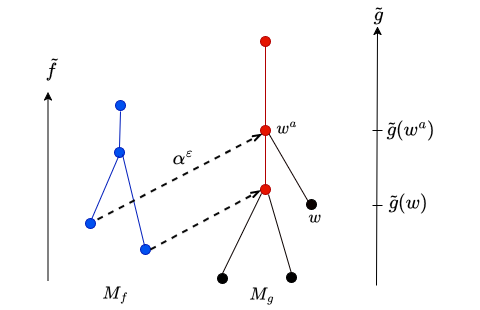}
        \caption{
           Illustration of the Ancestor-Closeness property in Definition~\ref{def:epsilon-goodmap} for an $\varepsilon$-good map \( \alpha^\varepsilon: M_f \rightarrow M_g \). The black region in \( M_g \) represents the complement of \( \mathrm{Im}(\alpha^{\varepsilon}) \) in $M_g$. For a point \( w \notin \mathrm{Im}(\alpha^{\varepsilon}) \), let \( w^a \) be its nearest ancestor such that \( w^a \in \mathrm{Im}(\alpha^{\varepsilon}) \). Then  we must have \( |\tilde{g}(w^a) - \tilde{g}(w)| \leq 2\varepsilon \).
        }
        \label{fig:ancestor-closeness}
    \end{minipage}
\end{figure}

 This reformulation enables \cite{touli2022fpt}  a dynamic programming approach by discretizing function values and constructing level structures that preserve hierarchical relations across both trees. 
 For a given candidate $\varepsilon > 0$, the  first step of the algorithm in \cite{touli2022fpt} involves augmenting the merge trees $M_f$ and $M_g$ by inserting additional degree-two nodes corresponding to appropriate levels of \( M_f \) and \( M_g \).
The collection of levels \( \mathcal{L}^{(\varepsilon)}_1 \) for \( M_f \) and \( \mathcal{L}^{(\varepsilon)}_2 \) for \( M_g \) are defined as follows:

\begin{align*}
    \mathcal{L}^{(\varepsilon)}_1 &= \{\tilde{f}^{-1}(c) \mid c \text{ is the } \tilde{f}\text{-image of a node of } M_f\} \cup \{\tilde{f}^{-1}(\hat{c} - \varepsilon) \mid \hat{c} \text{ is the } \tilde{g}\text{-image of a node of } M_g\}, \\
\mathcal{L}^{(\varepsilon)}_2 &= \{\tilde{g}^{-1}(c + \varepsilon) \mid c \text{ is the } \tilde{f}\text{-image of a node of } M_f\} \cup \{\tilde{g}^{-1}(\hat{c}) \mid \hat{c} \text{ is the } \tilde{g}\text{-image of a node of } M_g\}.   
\end{align*}
\begin{figure}[h]
    \centering
    \includegraphics[width=0.7\linewidth]{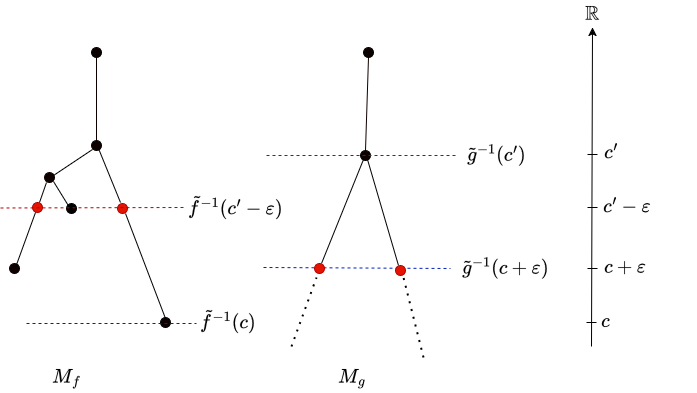}
     \caption{
        Illustration of the augmentation process. The red nodes denote newly inserted nodes, while the black nodes represent the already existing nodes of the merge trees $M_f$ and $M_g$.  Consider a critical level \( \tilde{f}^{-1}(c) \) in \( M_f \), then there is a corresponding level \( \tilde{g}^{-1}(c + \varepsilon) \) added in \( M_g \), and two additional red nodes of degree 2 are introduced at this level. Similarly, for a critical level \( \tilde{g}^{-1}(c') \) in \( M_g \), a corresponding level \( \tilde{f}^{-1}(c' - \varepsilon) \) is added in \( M_f \) and two additional degree two nodes are introduced. 
    }
    \label{fig:Augmentation}
\end{figure}
Note that there is a one-to-one correspondence between the levels in $\mathcal{L}^{(\varepsilon)}_1$ and the levels in $\mathcal{L}^{(\varepsilon)}_2$. 
These levels are sorted in ascending order by their function values. Let the $\tilde{f}$-values of the levels in \( \mathcal{L}^{(\varepsilon)}_1 \), associated with \( M_f \), be \( h_1, h_2, \dots, h_m \). Then the  $\tilde{g}$-values  of the corresponding levels in \( \mathcal{L}^{(\varepsilon)}_2 \), associated with \( M_g \), are \( \hat{h}_1, \hat{h}_2, \dots, \hat{h}_m \) where \( \hat{h}_i = h_i + \varepsilon \) for each $ i \in \{1, 2, \ldots, m \}$. That is, for every level \( \ell_i \in \mathcal{L}^{(\varepsilon)}_1 \) associated with \( M_f \) with  $\tilde{f}$-value \( h_i \), there is a corresponding level $\hat{\ell}_i\in \mathcal{L}^{(\varepsilon)}_2$ associated with \( M_g \) with $\tilde{g}$-value \( h_i + \varepsilon \), as illustrated in Figure~\ref{fig:Augmentation}. Finally, the merge trees $M_f$ and $M_g$ are augmented by introducing additional degree-two nodes at all the points of each level of  $\mathcal{L}^{(\varepsilon)}_1$ and $\mathcal{L}^{(\varepsilon)}_2$ where no node already exists. The augmented merge trees of $M_f$ and $M_g$ are denoted by $\widetilde{M}_f$ and $\widetilde{M}_g$, respectively.

In the second step, the algorithm by \cite{touli2022fpt} checks whether there exists an $\varepsilon$-good map between the augmented merge trees $\widetilde{M}_f$ and $\widetilde{M}_g$, as defined in Definition~\ref{def:epsilon-goodmap}. To do so, the algorithm first verifies the existence of at least one \emph{feasible valid pair} $(S, w)$ for every level $\ell_i \in \mathcal{L}^{(\varepsilon)}_1$, where $S$ is a subset of nodes in $\ell_i \in \mathcal{L}^{(\varepsilon)}_1$ with cardinality at most $\tau$ and $w$ is a node in the corresponding level $\hat{\ell}_i \in \mathcal{L}^{(\varepsilon)}_2$. Here, \( \tau \) is the degree-bound parameter as  defined in \cite{touli2022fpt}.
The pair $(S,w)$ is \textit{valid} if all nodes in $S$ share a common ancestor in $\widetilde{M}_f$ at $\tilde{f}$-value $h_i + 2\varepsilon$  and $\g(w) = \f(x)+\varepsilon$ for any $x \in S$. This construction guarantees that the \textbf{Range-Shift} and \textbf{Ancestor-Shift} properties of Definition~\ref{def:epsilon-goodmap} are satisfied.  
Next, the algorithm determines whether each such valid pair is \emph{feasible} using the following recursive procedure:  

\begin{enumerate}[Step~1:]
    \item For the first level $\ell_1 \in \mathcal{L}^{(\varepsilon)}_1$, a valid pair $(S, w)$ is feasible if the subtree of $\widetilde{M}_g$ rooted at $w$ has height at most $2\varepsilon$. Here, the height of a tree is defined as the maximum difference in function values between any two nodes in the tree.   
    
    \item For a level $\ell_i \in \mathcal{L}^{(\varepsilon)}_1$ with $i > 1$, the valid pair $(S, w)$ is feasible if there exists a partition $\{S_1, S_2, \ldots, S_k\}$ of the children of $S$, corresponding to the children $\{w_1, w_2, \ldots, w_k\}$ of $w$, such that each pair $(S_i, w_i)$ is feasible. Here the children of \( S \) are the nodes in \( \ell_{i-1} \) that are adjacent to the nodes in \( S \); similarly, the children of \( w \) are the nodes in \( \hat{\ell}_{i-1} \) that are adjacent to \( w \). This feasibility check is applied recursively.  
\end{enumerate}  

The feasibility of a valid pair $(S, w)$ ensures that the \textbf{Ancestor-Closeness} property of Definition~\ref{def:epsilon-goodmap} is satisfied. In particular, the feasibility of the pair $(\mathrm{root}(\widetilde{M}_f), \mathrm{root}(\widetilde{M}_g))$ guarantees the existence of an $\varepsilon$-good map between $M_f$ and $M_g$. If such a map exists, we say that $M_f$ and $M_g$ are $\varepsilon$-interleaved.  

Each $\varepsilon$ is chosen from a finite set of candidate values, as described in \cite{agarwal_computing_2018}. The smallest $\varepsilon$ for which $M_f$ and $M_g$ are $\varepsilon$-interleaved is the interleaving distance $\varepsilon^* = d_I(M_f, M_g)$. This candidate set is given in the following lemma:  

\begin{lemma}[\textbf{Candidate Set Generation}~\cite{agarwal_computing_2018}]  
    The interleaving distance $\varepsilon^* = d_I(M_f,M_g)$ belongs to the set $\Pi = \{|\tilde{f}(u) - \tilde{g}(v)| : u\in V(M_f), v\in V(M_g)\} \cup \{|\tilde{f}(u) - \tilde{f}(u')|/2 : u,u'\in V(M_f)\} \cup \{|\tilde{g}(v) - \tilde{g}(v')|/2 : v,v'\in V(M_g)\}$, where $V(M_f)$ and $V(M_g)$ denote the node sets of $M_f$ and $M_g$ respectively.
    \label{lemma:candidate values}
\end{lemma}  
The interleaving distance $\varepsilon^*=d_I(M_f,M_g)$ between two merge trees $M_f$ and $M_g$ is necessarily one of the values in the set $\Pi$ of size $O(n^2)$ where \( n \) is the total number of nodes in both \( M_f \) and \( M_g \). The algorithm in \cite{touli2022fpt} performs a binary search on the set of candidate values in $\Pi$ and tests whether an $\varepsilon$-good map exists between the merge trees $M_f$ and $M_g$.  

In contrast, our approach tests the existence of an $\varepsilon$-good map between $M_f$ and $M_g$ for each candidate $\varepsilon \in \Pi$, but replaces the subset-based matching strategy with a more efficient path-based approach. Rather than exhaustively enumerating all valid pairs $(S,w)$ at each level, we focus on the leaf-to-root paths in the merge tree and employ a path-based matching scheme that substantially reduces the search space.  
The next section formalizes our theoretical results, which form the foundation for the design and correctness of our algorithm presented subsequently.

\section{Theoretical Contributions}
\label{sec:theory}

First, we compute the sorted set $\Pi$ of candidate values and perform a binary search over it. At each step, we select an $\varepsilon \in \Pi$ and verify whether there exists an $\varepsilon$-good map between the given merge trees $M_f$ and $M_g$. If such a map exists, we continue our search to the left half of  $\Pi$ with values less than $\varepsilon$ to locate the smallest such $\varepsilon$; otherwise, we proceed with the right half. To ensure the correctness of the binary search algorithm, we establish a property known as the monotonicity of the interleaving distance, as follows.

\subsection{Monotonicity Property of the Interleaving Distance} 
\label{subsec: Theory- Monotonicity Property}
If \( M_f \) and \( M_g \) are \( \varepsilon \)-interleaved, then they are also \( \varepsilon' \)-interleaved for every \( \varepsilon' > \varepsilon \). This monotonicity property, formalized in Lemma~\ref{lemma: monotonicity}, is crucial for establishing the correctness of our binary search procedure. Before proving it, we first show that any continuous map \( \alpha^\varepsilon \) satisfying the Range-Shift property is ancestor-preserving. 
\begin{lemma}
\label{Lemma: ancestor-preservation}
   Let \( \alpha^\varepsilon : M_f \rightarrow M_g \) be a continuous map satisfying the Range-Shift property. If \( x' \succeq x \) for \( x, x' \in M_f \), then \( \alpha^\varepsilon(x') \succeq \alpha^\varepsilon(x) \). Equivalently,  
\[
\alpha^\varepsilon \circ i^\delta \;=\; j^\delta \circ \alpha^\varepsilon \quad \text{for each } \delta > 0.
\]
\end{lemma}

\begin{proof}

Let \( x' \succeq x \) for \( x, x' \in M_f \). Then there exists a path \( P \) in \( M_f \) connecting \( x \) and \( x' \) such that \( \tilde{f}(x') \geq \tilde{f}(x) \). Since \( \alpha^\varepsilon \) is continuous, the image \( \alpha^\varepsilon(P) \) is a path in \( M_g \) connecting \( \alpha^\varepsilon(x) \) and \( \alpha^\varepsilon(x') \). Furthermore, as \( \alpha^\varepsilon \) satisfies the Range-Shift property and \( \tilde{f}(x') \geq \tilde{f}(x) \), it follows that \( \tilde{g}(\alpha^\varepsilon(x')) \geq \tilde{g}(\alpha^\varepsilon(x)) \). Therefore, we have \( \alpha^\varepsilon(x') \succeq \alpha^\varepsilon(x) \).
    
    Again let, \(  i^\delta(x) = x' \). Then $\f(x') = \f(x) + \delta$. Applying $\alpha^\varepsilon$, we obtain $\alpha^\varepsilon(i^\delta(x)) = \alpha^\varepsilon(x') $. By the Range-Shift property, we have
    \begin{align*}
        \g(\alpha^\varepsilon(x')) &= \f(x') + \varepsilon = \f(x) + \delta + \varepsilon,\\
        \g(j^\delta(\alpha^\varepsilon(x))) &= \g(\alpha^\varepsilon(x)) +\delta = \f(x)+ \varepsilon +\delta.
    \end{align*}

    Both \( \alpha^\varepsilon(x') \) and $j^\delta(\alpha^\varepsilon(x)) $ are the ancestor of $\alpha^\varepsilon(x)$ at $\g$-value $\f(x) + \varepsilon +\delta$. As this ancestor is unique, we have $j^\delta(\alpha^\varepsilon(x)) = \alpha^\varepsilon(x')=\alpha^\varepsilon(i^\delta(x))$. Therefore we have $\alpha^\varepsilon \circ i^\delta = j^\delta \circ \alpha^\varepsilon$.
\end{proof}
We now establish the monotonicity property in the following lemma.
\begin{lemma}
    Given two merge trees \( M_f \) and \( M_g \), if there exists an \( \varepsilon \)-good map \( \alpha^{\varepsilon} : M_f \rightarrow M_g \), then for any \( \varepsilon' > \varepsilon \), there also exists an \( \varepsilon' \)-good map \( \alpha^{\varepsilon'} : M_f \rightarrow M_g \).
    \label{lemma: monotonicity}
\end{lemma}
\begin{proof}
   Let \( \alpha^{\varepsilon} : M_f \to M_g \) be an \( \varepsilon \)-good map. Set \( \gamma = \varepsilon' - \varepsilon\). Define 
   \[
\alpha^{\varepsilon'}(x) = j^\gamma \big( \alpha^\varepsilon(x) \big), \quad \text{ for all } x \in M_f.
\]
 Note that  the map $\alpha^{\varepsilon'}$ is continuous.
 
\textbf{Proof of Range-Shift property.} 

Since $\alpha^\varepsilon$ satisfies Range-Shift property, we have
    \[
    \tilde{g}(\alpha^{\varepsilon'}(x)) =\tilde{g}(j^\gamma(\alpha^{\varepsilon}(x))) =  \tilde{g}(\alpha^{\varepsilon}(x)) + \gamma = \tilde{f}(x) + \varepsilon + \gamma = \tilde{f}(x) + \varepsilon'
    \]
   which shows that Range-Shift property holds for \( \alpha^{\varepsilon'}\). 
   
   \textbf{Proof of Ancestor-Shift property.} 
   
   Suppose \( \alpha^{\varepsilon'}(x_1) \succeq \alpha^{\varepsilon'}(x_2) \) for any \( x_1, x_2 \in M_f \), that is \( j^\gamma(\alpha^\varepsilon(x_1)) \succeq j^\gamma(\alpha^\varepsilon(x_2)) \). By Lemma~\ref{Lemma: ancestor-preservation}, this implies $ \alpha^{\varepsilon}(i^\gamma (x_1)) \succeq  \alpha^{\varepsilon}(i^\gamma (x_2))$. Since $\alpha^\varepsilon$ satisfies Ancestor-Shift property, we have
\begin{align*}
    i^{2\varepsilon}(i^\gamma(x_1)) &\succeq i^{2\varepsilon}(i^\gamma(x_2)), \text{ and hence}\\
    i^{2\varepsilon+\gamma}(x_1) &\succeq  i^{2\varepsilon+\gamma}(x_2).
\end{align*}
 Applying \( i^\gamma \), which is ancestor-preserving, yields 
 \[
i^{2\varepsilon'}(x_1) \succeq i^{2\varepsilon'}(x_2).
\]
This proves the Ancestor-Shift property for $\alpha^{\varepsilon'}$.

 \textbf{Proof of Ancestor-Closeness property.} 
 
 Let \( w \in M_g \setminus \Img(\alpha^{\varepsilon'}) \). We consider the following two cases.

Case 1: $w \in M_g \setminus \Img(\alpha^{\varepsilon})$. 

By the Ancestor-Closeness property of $\alpha^\varepsilon$, the nearest ancestor $w^a \in \Img(\alpha^{\varepsilon})$ of $w$ satisfies $|\tilde{g}(w^{a}) - \tilde{g}(w)| \leq 2\varepsilon$. Take an $x \in M_f$ such that $\alpha^\varepsilon(x) = w^a$. Then 
\[
\alpha^{\varepsilon'}(x) = j^\gamma ( \alpha^\varepsilon(x)) =  j^\gamma (w^a)
\]
is an ancestor of $w$ that belongs to $\Img(\alpha^{\varepsilon'})$. Furthermore, we have 
\begin{align*}
    |\g(\alpha^{\varepsilon'}(x) ) - \g(w)| &\leq |\g(\alpha^{\varepsilon'}(x) ) - \g(w^a)| + |\tilde{g}(w^{a}) - \tilde{g}(w)|\\
    &\leq |\g(j^\gamma(w^a)) - \g(w^a)| + 2\varepsilon\\
    &= \gamma+ 2\varepsilon\\
    &\leq 2\varepsilon',
\end{align*}

which shows that, for the nearest ancestor $w^{a'}$ of $w$ that belongs to $\Img(\alpha^{\varepsilon'})$, we have $|\tilde{g}(w^{a'}) - \tilde{g}(w)| \leq |\tilde{g}(\alpha^{\varepsilon'}(x)) - \tilde{g}(w)| \leq 2 \varepsilon'$. This shows that $\alpha^{\varepsilon'}$ satisfies the Ancestor-Closeness property in Case 1.

Case 2: $w \in \Img(\alpha^{\varepsilon})\setminus \Img(\alpha^{\varepsilon'})$. 

Take an $x\in M_f$ such that $\alpha^\varepsilon(x) = w$. By definition 
\[
\alpha^{\varepsilon'}(x) = j^\gamma (\alpha^\varepsilon(x)) = j^\gamma(w)
\]
 is an ancestor of $w$ that belongs to $\Img(\alpha^{\varepsilon'})$.  
 Furthermore, the ancestor \( j^\gamma(w) \in \Img(\alpha^{\varepsilon'}) \) of \( w \) satisfies, 
 \[
 |\g(j^\gamma(w)) - \g(w)| = \gamma \leq 2\varepsilon'.
 \]
  Note that the nearest ancestor $w^a \in \Img(\alpha^{\varepsilon'})$ of $w$ lies between $w$ and $j^\gamma(w)$. Therefore,
   \[
|\g(w^a)- \g(w)|\leq  |\g(j^\gamma(w)) - \g(w)| \leq 2\varepsilon'.
 \]
   
   Hence, by Cases 1 and 2, Ancestor-Closeness property is satisfied for \( \alpha^{\varepsilon'}\). 
   Therefore, \( \alpha^{\varepsilon'} : M_f \rightarrow M_g \) is an \( \varepsilon' \)-good map.
\end{proof}

 We note, since the exact interleaving distance \( \varepsilon^*\) belongs to $\Pi$ (Lemma~\ref{lemma:candidate values}), it follows  from Lemma~\ref{lemma: monotonicity} that for every $\varepsilon' \in \Pi$ with $\varepsilon' > \varepsilon^*$, there exists an $\varepsilon'$-good  map between $M_f$ and $M_g$. Next section discusses how we construct a continuous map between the merge trees which is required to verify the existence of an $\varepsilon$-good map between them.
 

 \subsection{Constructing a Continuous Map $\varphi$ between Merge Trees}
 \label{subsec: Theory- Constructing a Continuous Map}
 Let \( L_f \) and \( L_g \) denote the sets of leaf nodes of the merge trees \( M_f \) and \( M_g \), respectively. We select an $\varepsilon \in \Pi$ and augment the merge trees $M_f$ and $M_g$ to obtain $\Mf$ and $\Mg$, respectively, as discussed in Section~\ref{subsec: FPT algorithm}. Note that this augmentation does not alter the topology of the merge trees. Consider a function $\phi: L_f \rightarrow \Mg$ such that each leaf node  $u \in L_f$ is mapped to a node at the level $\g^{-1}(\f(u)+\varepsilon)$ of $\Mg$, that is $\phi$ satisfies Range-Shift property. For each leaf node $u_i \in L_f$, we extend the mapping $u_i \mapsto \phi(u_i)$ along the leaf-to-root path $P_i$ starting from $u_i \in \Mf$. This yields an extension $\phi_i : P_i \rightarrow P_i'$, where $P_i'$ is the path from $\phi(u_i)$ to the root of $\Mg$. The extension $\phi_i$ is defined so that it satisfies the Range-Shift property, namely $\g(\phi_i(x)) = \f(x) + \varepsilon$ for any point $x$ in $P_i$. By construction, each map $\phi_i$ is continuous. Now, consider two distinct leaf nodes $u_i$ and $u_j$  and their corresponding maps $\phi_i: P_i \rightarrow P_i'$ and $\phi_j: P_j \rightarrow P_j'$. We define the  map \( \phi_i \cup \phi_j: P_i\cup P_j \rightarrow P_i' \cup P_j' \) as
\begin{equation}
\label{eqn:gluing-2fns}
(\phi_i \cup \phi_j)(x)=\displaystyle\left\{\begin{array}{ll}
                  \phi_i(x), & \text{for } x\in P_i, \\
                  \phi_j(x), & \text{for } x\in P_j
                \end{array}
              \right. 
\end{equation}
provided $\phi_i$ and $\phi_j$ agree on the intersection  $P_i\cap P_j$ which ensures the well-definedness of the map $\phi_i \cup \phi_j$. Since $\phi_i$ and $\phi_j$ are individually continuous  by construction, the following lemma shows that $\phi_i \cup \phi_j$ is continuous whenever $P_i$ and $P_j$ are closed subsets of $P_i \cup P_j$.

\begin{lemma}[Gluing Lemma \cite{armstrong1983basic}]
\label{lem:gluing-lem}
If $P_i$ and $P_j$ are closed in $P_i \cup P_j$ and if $\phi_i$ and $\phi_j$ are continuous, then the map $\phi_i \cup \phi_j$ is continuous.
\end{lemma}

Thus, to guarantee the continuity of $\phi_i \cup \phi_j$, it remains to show that $P_i$ and $P_j$ are closed subsets of $P_i \cup P_j$.

\begin{lemma}
\label{lem:path-closed}
$P_i$ and $P_j$ are closed in $P_i\cup P_j$ with topology induced from the quotient space topology on $\Mf$.
\end{lemma}
\begin{proof}
A path between any two points \( u \) and \( u' \) in \( P_i \cup P_j \) is a continuous function \( \gamma : [0,1] \rightarrow P_i \cup P_j \) such that \( \gamma(0) = u \) and \( \gamma(1) = u' \). Since \( [0,1] \) is compact and the continuous image of a compact set is compact~\cite{armstrong1983basic}, the image \( \gamma([0,1]) \) is compact in \( P_i \cup P_j \). Therefore, each of $P_i$ and $P_j$ is compact in $P_i\cup P_j$. Moreover, since every CW complex is Hausdorff \cite{Book-Hatcher}, the merge tree $\Mf$ is Hausdorff and consequently the subspace \(P_i \cup P_j \) is also Hausdorff. Finally, in a Hausdorff space compact subsets are closed, so both \(P_i\) and \(P_j\) are closed in \(P_i \cup P_j\).
\end{proof}

The following lemma establishes that, in order to prove that the map 
\(\phi_i \cup \phi_j\) is well-defined and continuous, 
it is sufficient to verify that 
 $\phi_i(v) = \phi_j(v)$, where \(v\) denotes the least common ancestor of the leaf nodes \(u_i\) and \(u_j\).

\begin{lemma}
  Given a map $\phi:L_f \rightarrow \widetilde{M}_g$ and two distinct leaf nodes $u_i,u_j \in L_f$, let $\phi_i :P_i \rightarrow P_i'$ and $\phi_j :P_j \rightarrow P_j'$ denote the extensions of the maps $u_i\mapsto \phi(u_i)$ and $u_j\mapsto \phi(u_j)$ along their respective leaf-to-root paths. Let \( v  \) be the least common ancestor of \( u_i \) and \( u_j \) in \( \widetilde{M}_f \), denoted by $\mathrm{LCA}(u_i, u_j)$. Then, the map \( \phi_i \cup \phi_j \), as defined in (\ref{eqn:gluing-2fns}), is well-defined on \( P_i \cup P_j \) if and only if
$\phi_i(v) = \phi_j(v)$.
   Furthermore, \( \phi_i \cup \phi_j \) is continuous under the same condition.
\label{lemma:continuity}
\end{lemma}

\begin{proof}
Let us show \( \phi_i \cup \phi_j \) is well-defined if and only if
$\phi_i(v) = \phi_j(v)$. 

The necessary part follows trivially:
 since if \( \phi_i(v) \neq \phi_j(v) \), then \( \phi_i \cup \phi_j \) is not well-defined at \( v \in P_i\cap P_j \). So \( \phi_i(v) = \phi_j(v) \) for $v=\mathrm{LCA}(u_i,u_j)$.

To prove the sufficient part: Suppose \( \phi_i(v) = \phi_j(v) \) for $v=\LCA(u_i, u_j)$, and let \( v \) lie at level \( \ell_k \in \mathcal{L}^{(\varepsilon)}_1 \)  in the augmented merge tree \( \widetilde{M}_f \), with function value \textcolor{black}{\( \tilde{f}(v) = h_k \)} (as defined in Section~\ref{subsec: FPT algorithm}). Note that since the paths $P_i$ and $P_j$ in  $\Mf$ are unique, \( P = P_i \cap P_j \) is a unique path from \( v \) to the root of \( \widetilde{M}_f \). Correspondingly, the path \( P' \) from \( \phi_i(v) = \phi_j(v) \) to the root of \( \widetilde{M}_g \) is also unique. 
For any ancestor \( v' \) of \( v \) along the path \( P \), there exists exactly one point on \( P' \) with \( \g \)-value \( \f(v') + \varepsilon \) to which both \( \phi_i \) and \( \phi_j \) can map $v'$ while satisfying the Range-Shift property. Hence, \( \phi_i(v') = \phi_j(v') \) for every ancestor $v'$ of $v$ on $P$. Thus  \( \phi_i \cup \phi_j \) is well-defined.
Furthermore, using Lemma~\ref{lem:path-closed} and \textbf{Gluing Lemma}~\ref{lem:gluing-lem}, the map \( \phi_i \cup \phi_j \) is continuous on \( P_i \cup P_j \). 
\end{proof}

Note that Lemma~\ref{lemma:continuity} can be generalized for a finite family of maps $\phi_i:P_i\rightarrow P_i'$ for $i=1, 2, \ldots, \eta_f$. First we define $\varphi = \phi_1 \cup \phi_2 \cup \cdots \cup \phi_{\eta_f}: \Mf \rightarrow \Mg$ as 
$$(\phi_1 \cup \phi_2 \cup \cdots \cup \phi_{\eta_f})(x)=\displaystyle\left\{\begin{array}{ll}
                  \phi_1(x), & \text{for } x\in P_1, \\
                  \phi_2(x), & \text{for } x\in P_2,\\  \hspace{8pt} \vdots & \hspace{15pt} \vdots\\
                 \phi_{\eta_f}(x), & \text{for } x\in P_{\eta_f}
                \end{array}
              \right.$$
\noindent
provided every two distinct functions $\phi_i$ and $\phi_j$ agree on the intersection $P_i\cap P_j$ ($i, j=1, 2, \ldots, \eta_f$).
The theorem is formalized as follows.


\begin{theorem}
Let \( u_1, \dots, u_{\eta_f} \) be the leaf nodes of the augmented merge tree \( \Mf \). Then the  map  \( \varphi= \phi_1 \cup \phi_2 \cup \cdots \cup \phi_{\eta_f} :\Mf\rightarrow\Mg\) 
 is well-defined if and only if for every pair $(i,j)$ where $i \ne j$ and \( i, j \in \{1, \dots, \eta_f\} \), the maps \( \phi_i \) and \( \phi_j \) agree at the least common ancestor of \( u_i \) and \( u_j \); that is $\phi_i(v) = \phi_j(v)  \text{ for } v = \mathrm{LCA}(u_i, u_j)$. Moreover, $\varphi$ is continuous.
\label{theorem:continuity}
\end{theorem}
\begin{proof} To prove the necessary part, suppose $\phi_i(v) \ne \phi_j(v)$  for $v = \LCA(u_i,u_j)$ of some pair of leaf nodes $u_i$ and $u_j$. Then $\varphi$ is not well-defined at $v$. Hence  $\phi_i(v) = \phi_j(v)$ for all pairs of leaf nodes $u_i$ and $u_j$, where $v=\LCA(u_i,u_j)$.

For the converse, assume that for every pair $(i,j)$ where $i \ne j$ and \( i, j \in \{1, \dots, \eta_f\} \), we have \( \phi_i(v) = \phi_j(v) \), where \( v = \mathrm{LCA}(u_i, u_j) \). First, consider \( \phi_1 \cup \phi_2 \), which is well-defined on $P_1 \cap P_2$ by Lemma~\ref{lemma:continuity}. We extend this to the map \( \varphi \) by induction. Now consider  \( (\phi_1 \cup \phi_2) \cup \phi_3 \). To check whether $\varphi$ is well-defined, we need to verify that
\[
\phi_1 \cup \phi_2 = \phi_3 \quad \text{on} \quad (P_1 \cup P_2) \cap P_3.
\]

Let $P$ denote $(P_1 \cup P_2) \cap P_3$. The branching of \( P_3 \) can occur at a node that lies either on \( P_1 \) or on \( P_2 \). If the branching begins at a node on \( P_2 \), as illustrated in Figure~\ref{fig:InP2}, then the path \( P \) is entirely contained in \( P_2 \), and thus $P = P_2 \cap P_3$. In this case, \( \phi_1 \cup \phi_2 = \phi_2 \) on \( P \), and since \( \phi_2 = \phi_3 \) on \( P_2 \cap P_3 \), it follows that \( \phi_1 \cup \phi_2 = \phi_3 \) on \( P \). A symmetric argument applies if the branching occurs along \( P_1 \), in which case \( P \subseteq P_1 \) and the same conclusion holds. If the branching begins at a node in $P_1 \cap P_2$ as shown in Figure~\ref{fig:InP1AndP2}, then since $\phi_1 = \phi_2$ on $P_1 \cap P_2$, both earlier cases apply simultaneously. Therefore, \( \phi_1 \cup \phi_2 \cup \phi_3 \) is well-defined on \( P_1 \cup P_2 \cup P_3 \).
 Moreover, since the finite union of closed sets is closed, both \( P_1 \cup P_2 \) and \( P_3 \) are closed in \( P_1 \cup P_2 \cup P_3 \), and since \( \phi_1 \cup \phi_2 \) and \( \phi_3 \) are continuous, it follows from the Gluing Lemma~\cite{armstrong1983basic} that \( \phi_1 \cup \phi_2 \cup \phi_3 \) is continuous on \( P_1 \cup P_2 \cup P_3 \). By repeating this process inductively for all \( i \in \{1, 2, \dots, \eta_f\} \), we obtain the map $\varphi = \phi_1 \cup \phi_2 \cup \cdots \cup \phi_{\eta_f}$
that is well-defined, since we see \( \phi_1 \cup \phi_2 \cup \cdots \cup \phi_{\eta_f - 1} = \phi_{\eta_f} \) on  \( (P_1 \cup P_2 \cup \cdots \cup P_{\eta_f-1}) \cap P_{\eta_f} \)  using the pairwise agreement at the least common ancestors. Moreover, the map \( \varphi \) is continuous, since both \( \phi_1 \cup \phi_2 \cup \cdots \cup \phi_{\eta_f - 1} \) and \( \phi_{\eta_f}\) are continuous, and \( P_1 \cup P_2 \cup \cdots \cup P_{\eta_f - 1} \) and \( P_{\eta_f} \) are closed in \( P_1 \cup P_2 \cup \cdots  \cup P_{\eta_f} \), as finite unions of compact sets remain compact. Thus, by the Gluing Lemma~\cite{armstrong1983basic}, the union \( \varphi = \phi_1 \cup \phi_2 \cup \cdots \cup \phi_{\eta_f} \) is continuous on \( \widetilde{M}_f \).
\end{proof}

 \begin{figure}[htbp]
    \centering
    \begin{minipage}{7cm}
        \centering
        \includegraphics[width=4cm]{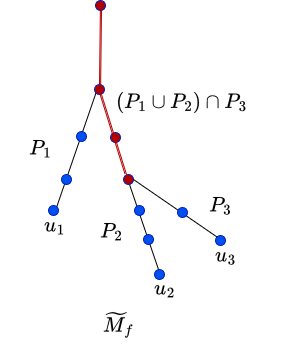}
\caption{Given an augmented merge tree $\widetilde{M}_f$, let $u_1,u_2$ and $u_3$ be the leaf nodes and $P_1,P_2$ and $P_3$ denote the corresponding leaf-to-root paths. $(P_1 \cup P_2) \cap P_3$ is coloured in red and it is contained in $P_2$.}

        \label{fig:InP2}
    \end{minipage}
    \hspace{2cm}
    \begin{minipage}{7cm}
      
        \centering
        \includegraphics[width=4cm]{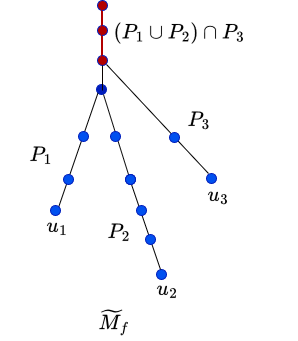}
        \caption{Given an augmented merge tree $\widetilde{M}_f$, let $u_1,u_2$ and $u_3$ be the leaf nodes and $P_1,P_2$ and $P_3$ denote the corresponding leaf-to-root paths. $(P_1 \cup P_2) \cap P_3$ is coloured in red and it is contained in $P_1$ and $P_2$.}
       \label{fig:InP1AndP2}
    \end{minipage}
  
\end{figure}

This result establishes the continuity of the map \( \varphi \)  on \( \widetilde{M}_f \) and the Range-Shift property holds by its construction.  It remains to check under what conditions $\varphi$ satisfies the Ancestor-Shift and Ancestor-Closeness properties of  Definition~\ref{def:epsilon-goodmap}.

\subsection{Ancestor-Shift Property of $\varphi$}
\label{subsec:thm-ancestor-shift}
In this section, we present a criterion for verifying that the constructed continuous map 
$\varphi : \widetilde{M}_f \to \widetilde{M}_g$ satisfies the Ancestor-Shift property. 
We begin by stating a necessary and sufficient condition for $\varphi$ to satisfy the Ancestor-Shift property, given in the following lemma.

\begin{lemma}
\label{lem:ancestor-shift-condition}
Let $\varphi : \widetilde{M}_f \to \widetilde{M}_g$ be the continuous map constructed in Theorem~\ref{theorem:continuity}. Then $\varphi$ satisfies the Ancestor-Shift property if and only if  
\begin{equation}
\label{eqn:ancestor-shift-condition}
\text{for any two points } x, y \in \widetilde{M}_f \text{ with } \f(x)=\f(y) \text{ and } \varphi(x) = \varphi(y) 
, \text{ we have}\;\;
i^{2\varepsilon}(x) = i^{2\varepsilon}(y).
\end{equation}
\end{lemma}

\begin{proof}
Suppose that $\varphi$ satisfies the Ancestor-Shift property. For two points 
$x, y \in \widetilde{M}_f$ with $\f(x) = \f(y)$ and $\varphi(x) = \varphi(y)$,
we have both $\varphi(x) \succeq \varphi(y)$ and $\varphi(y) \succeq \varphi(x)$. 
By the Ancestor-Shift property this implies 
$i^{2\varepsilon}(x) \succeq i^{2\varepsilon}(y)$ and 
$i^{2\varepsilon}(y) \succeq i^{2\varepsilon}(x)$, hence 
$i^{2\varepsilon}(x) = i^{2\varepsilon}(y)$. 
Therefore, condition~\eqref{eqn:ancestor-shift-condition} is necessary.

Conversely, suppose condition~\eqref{eqn:ancestor-shift-condition} holds. 
If $\varphi(x') \succeq \varphi(y')$ for some $x', y' \in \widetilde{M}_f$, 
then there exists an ancestor $y''$ of $y'$ with $\f(y'') = \f(x')$ such that 
$\varphi(x') = \varphi(y'')$. This is because $\varphi$ maps the monotone path from $y'$ to the root of $\Mf$ to that from $\varphi(y')$ to the root of $\Mg$, and the monotone path from $\varphi(y')$ to the root of $\Mg$ is unique. By condition~\eqref{eqn:ancestor-shift-condition}, 
we have $i^{2\varepsilon}(x') = i^{2\varepsilon}(y'')$. 
Since $i^{2\varepsilon}(y'') \succeq i^{2\varepsilon}(y')$, 
it follows that $i^{2\varepsilon}(x') \succeq i^{2\varepsilon}(y')$. 
Hence, $\varphi$ satisfies the Ancestor-Shift property.
\end{proof}

\textcolor{black}{The condition for the Ancestor-Shift property given in Lemma~\ref{lem:ancestor-shift-condition} 
can be further simplified by considering the leaf nodes of $\widetilde{M}_f$, 
which yields a more easily verifiable criterion for the constructed map 
$\varphi : \widetilde{M}_f \to \widetilde{M}_g$, as stated in the following theorem.}

\begin{theorem}
\label{thm:ancestor-shift-condition}
 Let $\varphi : \widetilde{M}_f \to \widetilde{M}_g$ be the continuous map constructed in 
Theorem~\ref{theorem:continuity}.  
For any pair of distinct leaf nodes $(u_i, u_j)$ in $\widetilde{M}_f$, let 
$v = \mathrm{LCA}(u_i, u_j)$ denote their least common ancestor, and let 
$P_i$ and $P_j$ be the leaf-to-root paths from $u_i$ and $u_j$, respectively.  
Traversing downward from $v$ along $P_i$ and $P_j$, let $u_r \in P_i$ and 
$u_s \in P_j$ be the first pair of nodes such that:  
\begin{enumerate}[(i)]
    \item $\tilde{f}(u_r) = \tilde{f}(u_s)$, and
    \item $\tilde{f}(v) - \tilde{f}(u_r) > 2\varepsilon$.
\end{enumerate}
If exists, such \textcolor{black}{a pair} of nodes $(u_r, u_s)$ will be referred to as $2\varepsilon$-pair corresponding to the pair of leaf nodes $(u_i, u_j)$. Then the map $\varphi$ satisfies the Ancestor-Shift property of 
Definition~\ref{def:epsilon-goodmap} if and only if
\[
   \varphi(u_r) \neq \varphi(u_s)
\]
for every pair of distinct leaf nodes $(u_i, u_j)$ in $\widetilde{M}_f$. Here, if for a given leaf pair $(u_i, u_j)$ either $u_r \in P_i$ or 
$u_s \in P_j$ (or both) does not exist, then we assume the condition 
$\varphi(u_r) \neq \varphi(u_s)$ is satisfied \textcolor{black}{formally}.

   \label{theorem:Condition (b)}
\end{theorem}

\begin{proof}
First, assume that the continuous map 
$\varphi: \widetilde{M}_f \rightarrow \widetilde{M}_g$ 
satisfies the Ancestor-Shift property of Definition~\ref{def:epsilon-goodmap}. 
Suppose, for contradiction, that $\varphi(u_r) = \varphi(u_s)$ for the \textcolor{black}{$2\varepsilon$-pair} $(u_r, u_s)$ 
of nodes corresponding to the leaf nodes $(u_i, u_j)$ with $v = \LCA(u_i, u_j)$ in $\widetilde{M}_f$. 
Since $\tilde{f}(v) - \tilde{f}(u_r) > 2\varepsilon$ (by condition (ii)), 
the nodes $u_r$ and $u_s$ cannot share a common ancestor at level 
$\tilde{f}(u_r) + 2\varepsilon$; that is, $i^{2\varepsilon}(u_r) \neq i^{2\varepsilon}(u_s)$. 
This contradicts the Ancestor-Shift property by Lemma~\ref{lem:ancestor-shift-condition}. 
Therefore, $\varphi(u_r) \neq \varphi(u_s)$, and the condition is necessary.

For the reverse direction, suppose that for every pair of distinct leaf nodes 
$(u_i, u_j)$ with $v = \mathrm{LCA}(u_i, u_j)$, the corresponding $2\varepsilon$-pair 
$(u_r, u_s)$ at level $\ell_k \in \mathcal{L}^{(\varepsilon)}_1$ satisfies 
$\varphi(u_r) \neq \varphi(u_s)$. 
We claim that $\varphi(x) = \varphi(y)$ cannot occur for any pair of ancestors 
$x, y$ of $u_r, u_s$, respectively, lying strictly between the levels $\ell_k$ and $\ell_{k+1}$ 
with equal $\tilde{f}$-values, nor for any pair of descendants $x, y$ of $u_r, u_s$, respectively, 
with equal $\tilde{f}$-values. 
This will be shown through the following two assertions.

\begin{assertion}
\label{assertion-1}
\( \varphi(x) \neq \varphi(y) \) for any ancestors  $x$  of  $u_r$ and $y$  of  $u_s$ 
 with  $\f(x) = \f(y)$, lying strictly between the levels \( \ell_k \) and \( \ell_{k+1} \).
\end{assertion}
    \begin{proof}
    Note that \textcolor{black}{$\tilde{f}(v) - h_{k+1} \leq 2\varepsilon$}, where 
\textcolor{black}{$\tilde{f}^{-1}(h_{k+1}) = \ell_{k+1}$}. 
Since there is no level in $\mathcal{L}^{(\varepsilon)}_1$ between $\ell_k$ and $\ell_{k+1}$, 
it follows that $\varphi(x)$ and $\varphi(y)$ must lie between the corresponding levels 
$\hat{\ell}_k$ and $\hat{\ell}_{k+1}$ in $\mathcal{L}^{(\varepsilon)}_2$, 
with no level of $\mathcal{L}^{(\varepsilon)}_2$ in between. 
Suppose, for contradiction, that $\varphi(x) = \varphi(y)$. 
Then this common point must be an ancestor of both $\varphi(u_r)$ and $\varphi(u_s)$, 
and hence the least common ancestor $v'$ of $\varphi(u_r)$ and $\varphi(u_s)$ has $\varphi(x) = \varphi(y)$ as an ancestor and $v'$ must be a node which has degree at least three in $\widetilde{M}_g$. 
Therefore $v'$ corresponds to a critical level, 
contradicting the assumption that there is no level of $\mathcal{L}^{(\varepsilon)}_2$ 
between $\hat{\ell}_k$ and $\hat{\ell}_{k+1}$. 
Therefore, $\varphi(x) \neq \varphi(y)$ for all ancestors $x$ and $y$ of $u_r$ and $u_s$, 
respectively, that share the same $\tilde{f}$-value and lie strictly between the levels 
$\ell_k$ and $\ell_{k+1}$.
\end{proof}

\begin{assertion}
\label{assertion-2}
    $\varphi(x) \neq \varphi(y)$ for any descendants $x$ of $u_r$ and $y$ of $u_s$ with $\tilde{f}(x) = \tilde{f}(y)$.
\end{assertion}
\begin{proof}
     Suppose, for contradiction, that $\varphi(x) = \varphi(y)$ for some descendants 
    $x$ of $u_r$ and $y$ of $u_s$ with $\tilde{f}(x) = \tilde{f}(y)$. 
    Since $\varphi$ is continuous and $\varphi(u_r) \neq \varphi(u_s)$, 
    the points $\varphi(x) = \varphi(y)$, $\varphi(u_r)$, $\varphi(u_s)$, and $\varphi(v)$ 
    together imply the existence of a loop in the merge tree $\widetilde{M}_g$, 
    which is impossible. 
    This contradiction, illustrated in Figure~\ref{fig:PairsExceeding2Epsilon}, 
    shows that $\varphi(x) \neq \varphi(y)$ for any such pair of descendants 
    with equal $\tilde{f}$-value.
\end{proof}

\textcolor{black}{In order to show that $\varphi$ satisfies the Ancestor-Shift property by using Lemma~\ref{lem:ancestor-shift-condition}, suppose that for $x, y \in \Mf$ with $\f(x) = \f(y)$, we have $\varphi(x) = \varphi(y)$.  If $x = y$, then we have $i^{2 \varepsilon}(x) = i^{2 \varepsilon}(y)$ trivially. Suppose $x \neq y$. Let $u_i$ and $u_j$ be the leaf nodes of $\Mf$ that are descendants of $x$ and $y$, respectively. If $u_i = u_j$, then, by the uniqueness of the leaf-to-root path from $u_i = u_j$, we must have $x = y$. which contradicts our assumption. Hence, $u_i \neq u_j$.} If for the pair $(u_i, u_j)$, a corresponding $2\varepsilon$-pair $(u_r, u_s)$ exists, 
then by Assertions~\ref{assertion-1} and~\ref{assertion-2}, the map $\varphi$ does not identify any ancestors of $u_r$ and $u_s$ 
lying between the levels $\ell_k$ and $\ell_{k+1}$ that share the same $\tilde{f}$-value, 
nor any of their descendants with the same $\tilde{f}$-value. 
Consequently, 
\textcolor{black}{$x$ and $y$ }must lie between the levels $\ell_{k+1}$ and $v$, implying $\tilde{f}(v) - \tilde{f}(x) \leq 2\varepsilon$. 
Since the path from $v$ to the root of $\widetilde{M}_f$ is unique, there exists an ancestor $v' \succeq v$ such that $\tilde{f}(v') - \tilde{f}(x) = 2\varepsilon$. 
As $v'$ is the unique common ancestor of both $x$ and $y$ satisfying 
$\tilde{f}(v') - \tilde{f}(x) = \tilde{f}(v') - \tilde{f}(y) = 2\varepsilon$, it follows that $i^{2\varepsilon}(x) = i^{2\varepsilon}(y)$. 
If either $u_r$ or $u_s$ (say, $u_r$) does not exist, then $\tilde{f}(v) - \tilde{f}(u_i) \leq 2\varepsilon$. 
Consequently, $\tilde{f}(v) - \tilde{f}(x) = \tilde{f}(v) - \tilde{f}(y) \leq 2\varepsilon$, 
and thus $i^{2\varepsilon}(x) = i^{2\varepsilon}(y)$ holds similarly as before. \textcolor{black}{Therefore, by Lemma~\ref{lem:ancestor-shift-condition}, the map $\varphi$ satisfies the Ancestor-Shift property.}
\end{proof}

\begin{figure}[h]
    \centering
    \includegraphics[scale=0.40]{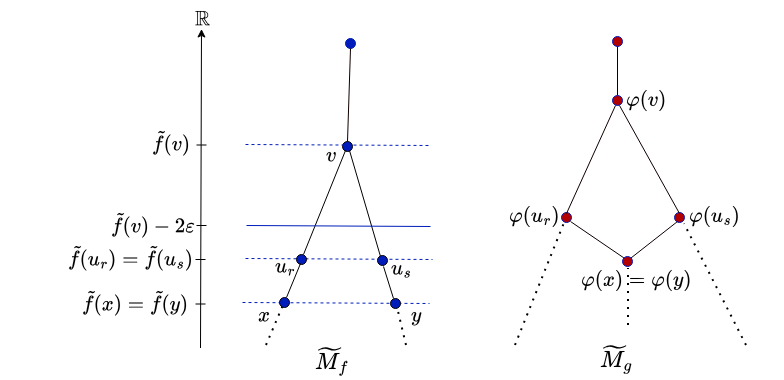}
    \caption{ Given two merge trees $\widetilde{M}_f$ and $\widetilde{M}_g$, let $v$ denote the least common ancestor of two leaf nodes in $\widetilde{M}_f$. Suppose $(u_r,u_s)$ is the pair corresponding to $v$ and let $x$ be a descendant of $u_r$ and $y$ be a descendant of $u_s$. If the map $\varphi : \widetilde{M}_f \rightarrow \widetilde{M}_g$ is continuous and satisfy $\varphi(x) = \varphi(y)$ and $\varphi(u_r) \ne \varphi(u_s)$, then $\varphi(v)$, $\varphi(u_r) $, $\varphi(u_s)$ and $\varphi(x) = \varphi(y)$ create a loop in $\widetilde{M}_g$.  }
    \label{fig:PairsExceeding2Epsilon}
\end{figure}

The continuous map $\varphi:\widetilde{M}_f \rightarrow \widetilde{M}_g$ satisfies the Range-Shift property and Ancestor-Shift property, and it remains to check the Ancestor-Closeness property for the map $\varphi$ to be an $\varepsilon$-good map.

\subsection{Ancestor-Closeness Property of $\varphi$} 
\label{subsec: Theory- Ancestor-Closeness}
In this section, to verify the Ancestor-Closeness property, we show that it suffices to consider 
the set $L_g$ of all leaf nodes of $\widetilde{M}_g$. 
For any leaf node $w \in L_g \setminus \mathrm{Im}(\varphi)$, 
let $w^a$ denote the nearest ancestor of $w$ that lies in $\mathrm{Im}(\varphi)$. 
We then need to check that $|\tilde{g}(w^a) - \tilde{g}(w)| \leq 2\varepsilon$. 
Since leaf nodes are extremal in the tree structure and $\varphi$ is monotone by construction, 
verifying the Ancestor-Closeness property only for the leaf nodes 
$w \in L_g \setminus \mathrm{Im}(\varphi)$ is sufficient to guarantee that 
the property holds for all other points in $\widetilde{M}_g \setminus \mathrm{Im}(\varphi)$. 
We formalize this observation in the following lemma.

\begin{lemma}
 Let $\varphi:\widetilde{M}_f \rightarrow \widetilde{M}_g$ be the continuous map in Theorem~\ref{theorem:continuity}. If $\varphi$ satisfies Ancestor-Closeness property for all leaf nodes $w \in L_g\setminus \operatorname{Im}(\varphi)$, then $\varphi$ satisfies the Ancestor-Closeness property for all point $x \in \Mg \setminus \operatorname{Im}(\varphi)$. 
\label{lemma:ancestor-closeness}
\end{lemma}

\begin{proof}
Suppose Ancestor-Closeness property of Definition~\ref{def:epsilon-goodmap} holds for all leaf nodes \( w \in L_g \setminus \operatorname{Im}(\varphi) \). We aim to show that for any point \( x \in \widetilde{M}_g \setminus \operatorname{Im}(\varphi) \), the inequality  
\[
|\tilde{g}(x) - \tilde{g}(x^a)| \leq 2\varepsilon
\]  
holds, where \( x^a \) denotes the nearest ancestor of \( x \) that lies in \( \operatorname{Im}(\varphi) \).

We claim that any such point \( x \in \widetilde{M}_g \setminus \operatorname{Im}(\varphi) \) must lie on the path  connecting the unique leaf node \( w \in L_g \setminus \operatorname{Im}(\varphi) \), which is a descendant of $x$, to its nearest ancestor \( w^a \in \operatorname{Im}(\varphi) \). Otherwise, \( x \succeq w^a\).
Since \( \varphi \) is continuous and monotone, it follows that 
\( x \in \operatorname{Im}(\varphi) \). 
This contradicts the assumption that 
\( x \notin \operatorname{Im}(\varphi) \).

Therefore, \( x \) must lie between some \( w \in L_g \setminus \operatorname{Im}(\varphi) \) and its nearest ancestor \( w^a \in \operatorname{Im}(\varphi) \). In this case, the nearest ancestor of \( x \) in \( \operatorname{Im}(\varphi) \) is precisely \( x^a = w^a \). Since all leaf nodes \( w \in L_g \setminus \operatorname{Im}(\varphi) \) satisfies $|\tilde{g}(w) - \tilde{g}(w^a)| \leq 2\varepsilon$, we have:
\[
|\tilde{g}(x) - \tilde{g}(x^a)| \leq |\tilde{g}(w) - \tilde{g}(w^a)| \leq 2\varepsilon.
\]
Hence, Ancestor-Closeness property is satisfied for every point \( x \in \widetilde{M}_g \setminus \operatorname{Im}(\varphi) \).
\end{proof}









 Lemma~\ref{lemma:ancestor-closeness} establishes that it is sufficient to verify the Ancestor-Closeness property of Definition~\ref{def:epsilon-goodmap} only for the leaf nodes \( w \in L_g \setminus \operatorname{Im}(\varphi) \). Once all three properties are verified, we determine that \( \varphi \) qualifies to be an \( \varepsilon \)-good map. Our goal is to compute the smallest such value of \( \varepsilon \), denoted \( \varepsilon^* \), for which an \( \varepsilon \)-good map exists. We now proceed to describe the algorithmic framework used to compute \( \varepsilon^* \).

\section{Algorithmic Contributions}
\label{Sec:algorithm}
A key structural property of merge trees is that each tree admits a unique decomposition into monotone paths from its leaves to the root, and the union of all these paths covers  the tree. Leveraging this property, we avoid the need to examine all nodes and their subsets when determining the existence of an \( \varepsilon \)-good map between two merge trees. Instead, we focus solely on the leaf nodes of the first merge tree and construct maps along each corresponding leaf-to-root path. Let \( M_f \) and \( M_g \) be the input merge trees. Let \( \widetilde{M}_f \) and \( \widetilde{M}_g \) denote their corresponding augmented versions, as defined in Section~3.3. Let \( L_f \) and \( L_g \) denote the sets of leaf nodes in \( \widetilde{M}_f \) and \( \widetilde{M}_g \), respectively. We denote by $\eta_f$ and $\eta_g$ the numbers of leaf nodes of $\Mf$ and $\Mg$ respectively. Here, we assume that the number of possible $\varepsilon$-good maps from $M_f$ to $M_g$ does not exceed that from $M_g$ to $M_f$, i.e., $\eta_g^{\eta_f} \leq \eta_f^{\eta_g}$, in order to reduce the complexity of our algorithm. If $\eta_g^{\eta_f} > \eta_f^{\eta_g}$, we construct the $\varepsilon$-good maps 
from $M_g$ to $M_f$ in an analogous manner, interchanging the roles of 
$M_f$ and $M_g$.  Note that the number of leaf-to-root paths in each merge tree equals the number of leaves. Our algorithm computes the exact interleaving distance \( \varepsilon^* = d_I(M_f, M_g) \) by performing the following steps:

\begin{enumerate}[Step~1:]
    \item \textbf{Candidate Value Generation.}
We begin by constructing a sorted list \( \Pi \) of candidate values \( \varepsilon \), obtained by considering all pairwise differences in function values between nodes across the merge trees \( M_f \) and \( M_g \), as well as within each individual tree, as in Lemma~3.1. The exact interleaving distance \( \varepsilon^* \) is identified as the smallest value in \( \Pi \) for which the merge trees \( M_f \) and \( M_g \) are \( \varepsilon \)-interleaved. To determine \( \varepsilon^* \), we perform a binary search over the sorted list \( \Pi \). For the correctness of binary search algorithm, we have proved that if $M_f$ and $M_g$ are \( \varepsilon \)-interleaved,  then for every \( \varepsilon' > \varepsilon \), $M_f$ and $M_g$ are \( \varepsilon' \)-interleaved (as in Lemma \ref{lemma: monotonicity}). At each iteration of the binary search, a candidate value \( \varepsilon \in \Pi \) is selected, and we check whether \( M_f \) and \( M_g \) are \( \varepsilon \)-interleaved. This verification proceeds through the following steps. Section~\ref{Subsec:Candidate set generation} describes the procedure for generating the set of candidate values.

\item \textbf{Extending the Merge Trees.}
\textcolor{black}{As discussed in Section~\ref{subsec: Background - Interleaving}, 
for a given candidate value $\varepsilon$, we first extend the merge trees $M_f$ and $M_g$ 
by modifying the root of $M_f$ or $M_g$, if necessary, 
to facilitate the definition of an $\varepsilon$-good map between them. 
Details of the extension procedure are provided in Section~\ref{Subsubsec:Extending-Merge-Trees}.
}

    \item \textbf{Augmenting the Merge Trees.}
Next, we augment the extended merge trees \( M_f \) and \( M_g \) by inserting degree-two nodes at each levels of \( \mathcal{L}^{(\varepsilon)}_1 \) in \( M_f \) and \( \mathcal{L}^{(\varepsilon)}_2 \) in \( M_g \) where no node already exists. Then we obtain the augmented merge trees \( \widetilde{M}_f \) and \( \widetilde{M}_g \), as described in Section~\ref{subsec: FPT algorithm}. For details of the augmentation procedure, refer to Section~\ref{Subsubsec:Augmenting-Merge-Trees}.

    \item \textbf{Constructing Maps between Augmented Merge Trees.}
For a given candidate value \( \varepsilon \), we begin by constructing a map \( \phi: L_f \rightarrow \widetilde{M}_g \), where each leaf node in $L_f$ of \( \widetilde{M}_f \) is mapped to a node at the corresponding level in \( \widetilde{M}_g \). This initial assignment is then extended along each leaf-to-root path in \( \widetilde{M}_f \) to obtain a map \( \varphi: \widetilde{M}_f \rightarrow \widetilde{M}_g \). Specifically, for each \( u_i \in L_f \), we extend the mapping \( u_i \mapsto \phi(u_i) \) to a map \( \phi_i: P_i \rightarrow P_i' \), where \( P_i \) denotes the path from \( u_i \) to the root of \( \widetilde{M}_f \), and \( P_i' \) is the corresponding path from \( \phi(u_i) \) to the root of \( \widetilde{M}_g \). Each node on \( P_i \) is mapped to the unique node on \( P_i' \) that lies at the corresponding level. The complete map \( \varphi \) is considered as the union \( \varphi = \bigcup_{i=1}^{\eta_f} \phi_i \) provided it is well-defined, i.e. each node in $\widetilde{M}_f$ has a unique image in $\widetilde{M}_g$ (as described in Section~\ref{sec:theory}).
Furthermore, we have shown that $\varphi$ is continuous in
Theorem~\ref{theorem:continuity}.
In the construction of $\varphi$, a node \( u \) on a path \( P \subseteq \widetilde{M}_f \) at level \( \ell_k \in \mathcal{L}^{(\varepsilon)}_1 \), with function value \textcolor{black}{\( \tilde{f}(u) = h_k \)}, is mapped to the unique node on the corresponding path \( P' \subseteq \widetilde{M}_g \) that lies at level \( \hat{\ell}_k \in \mathcal{L}^{(\varepsilon)}_2 \), with function value \textcolor{black}{\( \tilde{g}(\varphi(u)) = \hat{h}_k = h_k + \varepsilon \)}. This construction guarantees that \textbf{Range-Shift} property of Definition~\ref{def:epsilon-goodmap} is satisfied. We note that since every node at a level lies on some leaf-to-root path, the number of nodes at any level in the augmented tree \( \widetilde{M}_g \) is bounded above by \( \eta_g \), the number of leaf nodes of $\Mg$. As a result, each leaf node in \( \widetilde{M}_f \) has at most \( \eta_g \) target nodes in \( \widetilde{M}_g \), leading to  at most \( \eta_g^{\eta_f} \) such possible mappings $\varphi$. The details of our map construction algorithm is described in Section~\ref{subsec:constructing-maps}.

   \item \textbf{Verifying $\varepsilon$-Goodness of Map.} Finally, we verify if a map \( \varphi \) between the augmented merge trees \( \widetilde{M}_f \) and \( \widetilde{M}_g \) qualifies to be an $\varepsilon$-good map. It remains to verify whether \( \varphi \) also satisfies \textbf{Ancestor-Shift} and \textbf{Ancestor-Closeness} properties of Definition~\ref{def:epsilon-goodmap}.  To verify \textbf{Ancestor-Shift} property, we have proved  Theorem~\ref{theorem:Condition (b)}, which ensures that if \( \varphi(x) = \varphi(y) \) for some distinct nodes \( x, y \in \widetilde{M}_f \), then \( x \) and \( y \) must share a common ancestor at level \( \tilde{f}(x) + 2\varepsilon \). For \textbf{Ancestor-Closeness} property, Lemma~\ref{lemma:ancestor-closeness} shows that it suffices to verify this property only for the leaf nodes in the set \( L_g \setminus \mathrm{Im}(\varphi) \). This  verification ensures that \textbf{Ancestor-Closeness} property is satisfied for the entire map \( \varphi: \widetilde{M}_f \rightarrow \widetilde{M}_g \). See Section~\ref{subsubsec: Checking epsilon-goodness} for details of the algorithm.
\end{enumerate}

Following these steps, we now describe our main algorithm \textsc{ComputeInterleavingDistance} in Algorithm \ref{alg:main} which computes the interleaving distance between two merge trees $M_f$ and $M_g$. Beginning with Step 1, the procedure \textsc{GenerateCandidateValues} generates the sorted list  \( \Pi \) of candidate values (line 3, Algorithm~\ref{alg:main}). A binary search is performed to find the exact interleaving distance $\varepsilon^*$ between the merge trees (lines $6$-$15$, Algorithm~\ref{alg:main}).  Two indices, $l$ and $r$, are initialized 
to the ends of this list, and $\varepsilon^*$ is set to infinity representing the current best interleaving value found so far (lines $4$-$5$, Algorithm~\ref{alg:main}). At each iteration, a candidate value \( \varepsilon \in \Pi \) is selected, and the procedure \textsc{Is-$\varepsilon$-Interleaved} is invoked to check whether the merge trees \( M_f \) and \( M_g \) are \( \varepsilon \)-interleaved (line 9, Algorithm~\ref{alg:main}). The procedure \textsc{Is-$\varepsilon$-Interleaved} consists of the Steps:  2. Extending the merge trees 3. Augmenting the merge trees, 4. Constructing maps between augmented merge trees and 5. Verifying the $\varepsilon$-goodness of map. See Section~\ref{subsec: Is-epsilon-interleaved} for details of the \textsc{Is-$\varepsilon$-Interleaved} procedure.  If \( M_f \) and \( M_g \) are \( \varepsilon \)-interleaved, then the search continues in the left half of \( \Pi \) considering values less than \( \varepsilon \) (lines $9$-$11$, Algorithm~\ref{alg:main}).
 Otherwise, the search proceeds in the right half of \( \Pi \), with values greater than \( \varepsilon \) (lines $12$-$13$, Algorithm~\ref{alg:main}).
The process 
repeats until the range is exhausted, at which point the smallest valid 
$\varepsilon$ found is returned as the exact interleaving distance  \( \varepsilon^*\) (line 16, Algorithm~\ref{alg:main}).
\begin{algorithm}[H]
\caption{\textsc{ComputeInterleavingDistance}($M_f, M_g$)}
\label{alg:main}
\begin{algorithmic}[1]
\State \textbf{Input:} Merge trees $M_f$ and $M_g$ 
\State \textbf{Output:} $\varepsilon^* = d_I(M_f, M_g)$
\State $\Pi \gets \textsc{GenerateCandidateValues}(M_f, M_g)$ 
\State $l \gets 0$, $r \gets \Pi.\mathrm{length}() - 1$
\State $\varepsilon^* \gets \infty$
\While{$l \leq r$}
    \State $mid \gets \lfloor (l + r)/2 \rfloor$
    \State $\varepsilon \gets \Pi.\mathrm{get}(mid)$
    \If{\textsc{Is-$\varepsilon$-Interleaved}$(M_f, M_g, \varepsilon )$}
        \State $\varepsilon^* \gets \varepsilon$
        \State $r \gets mid - 1$
    \Else
        \State $l \gets mid + 1$
    \EndIf
\EndWhile
\State \Return $\varepsilon^*$
\end{algorithmic}
\end{algorithm}
Next, we describe all the \textcolor{black}{five} steps of the algorithm in details. 

\subsection{Generation of Candidate Values }
\label{Subsec:Candidate set generation}
The procedure \textsc{GenerateCandidateValues}  enumerates all possible function value differences both between nodes within the same tree and across the trees \( M_f \) and \( M_g \), and stores them in the set \( \Pi \) \textcolor{black}{(lines 5--21, procedure \textsc{GenerateCandidateValues})}. \textcolor{black}{Inserting into a set data-structure ensures that no duplicates are inserted \textcolor{black}{(lines 8, procedure \textsc{GenerateCandidateValues})}}. Then procedure \textsc{Sort}  converts the set \( \Pi \)  into a  sorted list of elements in increasing order to enable efficient binary search in Algorithm~\ref{alg:main} \textcolor{black}{(lines 22--23, procedure \textsc{GenerateCandidateValues})}.
The  procedure \textsc{GenerateCandidateValues} is detailed below.

\begin{algorithmic}[1]
\Procedure{\textsc{GenerateCandidateValues}}{$M_f, M_g$}
\label{proc:GenerateCandidateValues}
    \State Initialize set: $\Pi \gets \emptyset$;
    \State \textcolor{black}{$V_f \gets M_f.\textsc{GetNodes}()$};  \State $V_g\gets M_g.\textsc{GetNodes}()$;
    \State \% \textit{Function value differences for all node pairs across the merge trees.}
    \ForAll{$u \in V_f$}
        \ForAll{$v \in V_g$}
            \State $\Pi .\mathrm{insert}\left( |\tilde{f}(u) - \tilde{g}(v)| \right)$
        \EndFor
    \EndFor
 \State \% \textit{Function value differences for all node pairs within the merge trees.}
    \For{$u \in V_f$}
    \For{$u' \in V_f$ $\And$ $u \ne u'$}
        \State $\Pi.\mathrm{insert}\left( \frac{|\tilde{f}(u) - \tilde{f}(u')|}{2} \right)$
    \EndFor    
    \EndFor

    \For{$v \in V_g$}
    \For{$v' \in V_g$ $\And$ $v \ne v'$}
        \State $\Pi .\mathrm{insert}\left( \frac{|\tilde{g}(v) - \tilde{g}(v')|}{2} \right)$
    \EndFor    
    \EndFor
    \State \% \textit{Return list of sorted values without any duplicate entries}
    \State \Return{$\textsc{Sort}(\Pi)$}
\EndProcedure
\end{algorithmic}

Once the sorted list $\Pi$ of candidate values has been generated, we pick a value $\varepsilon \in \Pi$ and test whether $M_f$ and $M_g$ are $\varepsilon$-interleaved. The process starts by augmenting each merge tree as detailed in the next section.

\subsection{Extending the Merge Trees.}
\label{Subsubsec:Extending-Merge-Trees}
This section describes the  procedure for extending the  merge trees. In Section~\ref{subsec: Background - Interleaving}, for the well-definedness of
the $\varepsilon$-compatible maps between two merge trees, we extended each of the merge trees with a ray upward from the root, allowing the
function value to tend to $+\infty$.  However, for implementation purposes, to compute $\varepsilon$-good map, instead of extending a ray, we modify the merge trees \( M_f \) and \( M_g \) as follows. \color{black} If the $\g$-value at the root of \( M_g \) satisfies \( \g(\mathrm{root}(M_g)) < \f(\mathrm{root}(M_f)) + \varepsilon \), a new root node is created in \( M_g \) with $\g$-value \( \f(\mathrm{root}(M_f)) + \varepsilon \), and the original root of \( M_g \) is assigned as its child. Conversely, if \( \g(\mathrm{root}(M_g)) > \f(\mathrm{root}(M_f)) + \varepsilon \), the procedure extends \( M_f \) in a similar manner by creating a new root node with  
$\f$-value $\g(\mathrm{root}(M_g)) - \varepsilon$, and the original root of $M_f$ is made its child.
The pseudocode for this extension process is provided in the procedure \textsc{ExtendMergeTree}, which is invoked prior to the augmentation procedure.


\begin{algorithmic}[1]
\Procedure{\textsc{ExtendMergeTree}}
{$M_g, M_f, \varepsilon$}
\label{proc:ExtendMergeTree}
\State $t \gets \f(\mathrm{root}(M_f)) + \varepsilon$
\If{$\g(\mathrm{root}(M_g)) < t$}
    \State $r_{\text{old}} \gets \mathrm{root}(M_g)$
    \State $r_{\text{new}} \gets$ \textsc{CreateNode}()
    \State $\g(r_{\text{new}}) \gets t$
    \State $r_{\text{new}}.\textsc{AddChild}(r_{\text{old}})$
    \State $M_g.\textsc{SetRoot}(r_{\text{new}})$
  \color{black}  \ElsIf{$\g(\mathrm{root}(M_g)) > t$}
    \State $r_{\text{old}} \gets \mathrm{root}(M_f)$
    
        \State $r_{\text{new}} \gets$ \textsc{CreateNode}()
        \State $\f(r_{\text{new}}) \gets \g(\mathrm{root}(M_g)) - \varepsilon$
        \State $r_{\text{new}}.\textsc{AddChild}(r_{\text{old}})$
    \State $M_f.\textsc{SetRoot}(r_{\text{new}})$
\EndIf
\State \Return $\{M_f, M_g\}$
\EndProcedure
\end{algorithmic}

\noindent
The procedure \textsc{ExtendMergeTree} ensures that the root of $M_g$ satisfies the condition 
$\g(\mathrm{root}(M_g)) = \f(\mathrm{root}(M_f)) + \varepsilon$. 
It first computes the target value $t = \f(\mathrm{root}(M_f)) + \varepsilon$ (line~2, procedure \textsc{ExtendMergeTree}). 
If the $\g$-value of the root of $M_g$ is smaller than $t$, a new root node $r_{\text{new}}$ is created 
and assigned the value $\g(r_{\text{new}}) = t$ (lines~3--6, procedure \textsc{ExtendMergeTree}). 
The original root $r_{\text{old}}$ of $M_g$ is then attached as a child of $r_{\text{new}}$ (line~7, procedure \textsc{ExtendMergeTree}), 
and $r_{\text{new}}$ is set as the new root of $M_g$ (line~8, procedure \textsc{ExtendMergeTree}), see the Figure~\ref{fig:Extension-of-merge-trees}(a). If $\g(\mathrm{root}(M_g)) > t $, then the root of $M_f$ is modified analogously (lines 9--15, procedure \textsc{ExtendMergeTree}), see the Figure~\ref{fig:Extension-of-merge-trees}(b).
Finally, the extended merge trees $M_f$ and $M_g$ are returned (line~16, procedure \textsc{ExtendMergeTree}).

\begin{figure}[h!]
\centering
\includegraphics[width=0.96\textwidth]{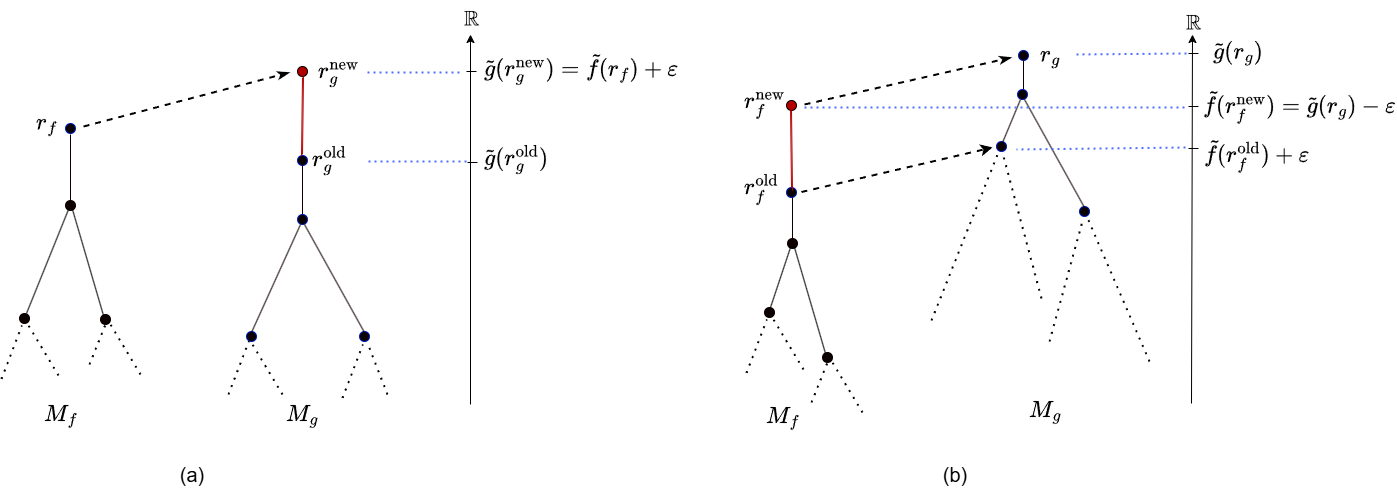}
\caption{(a) Illustration of extending the merge tree $M_g$ when $\g(r_g^{\text{old}}) < \f(r_f) + \varepsilon$. Here, $r_f$ and $r_g^{\text{old}}$ denote the original roots of $M_f$ and $M_g$, respectively. 
A new root $r_g^{\text{new}}$ is introduced in $M_g$ with $\g$-value $\f(r_f) + \varepsilon$. 
The original root $r_g^{\text{old}}$ is then assigned as the child of this newly created root $r_g^{\text{new}}$.
(b) Illustration of extending the merge tree $M_f$ when $\g(r_g) > \f(r_f^{\text{old}}) + \varepsilon$. 
Here, $r_f^{\text{old}}$ and $r_g$ denote the original roots of $M_f$ and $M_g$, respectively. 
A new root $r_f^{\text{new}}$ is introduced in $M_f$ with $\f$-value $\g(r_g) - \varepsilon$. 
The original root $r_f^{\text{old}}$ is then assigned as the child of this newly created root $r_f^{\text{new}}$.
\label{fig:Extension-of-merge-trees}}
\end{figure}

\color{black}

\subsection{Augmenting the Merge Trees.}
\label{Subsubsec:Augmenting-Merge-Trees}
Next, we describe the augmentation of the extended merge tree $M_f$; the augmentation of the extended merge tree $M_g$ is performed analogously. The procedure \textsc{AugmentMergeTree} which augments the merge tree $M_f$ to produce the augmented tree $\widetilde{M}_f$ is outlined below.

\begin{algorithmic}[1]
\Procedure{\textsc{AugmentMergeTree}}{$M_f, M_g, \varepsilon$}
\label{proc:AugmentMergeTree}
\State $\Mf \gets M_f$
  \State $V \gets \textcolor{black}{\Mf.\textsc{GetNodes}()}$
  \State $E \gets \Mf.\textsc{GetEdges}()$
  \State $\hat{V} \gets M_g.\textsc{GetNodes}()$
  
 \State \% \textit{Augment $M_f$ at all levels $\f(u)$ for $u \in V$}
  \For{$u \in V$}
    \For{$(x,y) \in E  ~\And ~ \f(x) > \f(y) $}
      \If{$\f(x) > \f(u) \ \And\ \f(y) < \f(u)$}
        \State $\Mf.\textsc{RemoveEdge}(x,y)$
        \State $z \gets \textcolor{black}{\Mf.\textsc{CreateNode}()}$
        \State $\f(z) \gets \f(u)$
        \State $\Mf. \textcolor{black}{\textsc{AddEdges}(\{(x,z), (z,y)\})}$
        \State $\Mf.\mathrm{update}()$
      \EndIf
    \EndFor
  \EndFor
  
  \State \% \textit{Augment $M_f$ at levels $\g(v) -\varepsilon$ for $v \in \hat{V}$}
  \For{$v \in \hat{V}$}
    \For{$(x,y)\in E  ~\And ~ \f(x) > \f(y) $}
      \If{$\f(x) > \g(v) -\varepsilon \ \And\ \f(y) < \g(v) - \varepsilon$}
        \State $\Mf.\textsc{RemoveEdge}(x,y)$
        \State $z \gets \Mf.\textsc{CreateNode}()$
        \State $\f(z) \gets \g(v) -\varepsilon$
        \State $\Mf.\textcolor{black}{ \textsc{AddEdges}(\{(x,z), (z,y)\})}$
        \State $\Mf.\mathrm{update}()$
      \EndIf
    \EndFor
  \EndFor
  \State \Return{$\Mf$}
  \EndProcedure
\end{algorithmic}

For each node $u \in M_f$ and each edge $(x, y)$ in $M_f$, if $\f(x) > \f(u) > \f(y)$, a new degree-two node $z$ is inserted at level $\f(u)$ by replacing the edge $(x, y)$ with $(x, z)$ and $(z, y)$ and setting $\f(z) = \f(u)$ \textcolor{black}{(lines~6--17, procedure \textsc{AugmentMergeTree})}. Subsequently, for a given candidate value $\varepsilon$ and for each node $v \in M_g$, 
if there exists an edge $(x, y)$ in $M_f$ such that 
$\f(x) > \g(v) - \varepsilon > \f(y)$, 
a new degree-two node $z$ is inserted at the level $\g(v) - \varepsilon$. 
This is achieved by splitting the edge $(x, y)$ into two edges, $(x, z)$ and $(z, y)$, 
and assigning $\f(z) = \g(v) - \varepsilon$ 
\textcolor{black}{(lines~18--29, procedure \textsc{AugmentMergeTree})}.
 Following the augmentation, the construction of mappings between the augmented merge trees is carried out in the next section.

\subsection{Constructing Maps between Augmented Merge Trees.} 
\label{subsec:constructing-maps}


The procedure \textsc{ConstructMap} extends a given leaf-to-node correspondence 
$\phi : L_f \rightarrow \widetilde{M}_g$ to a continuous map 
$\varphi : \widetilde{M}_f \rightarrow \widetilde{M}_g$ where construction of $\phi$ will be discussed in Section~\ref{subsec: Is-epsilon-interleaved}. 
The algorithm begins by initializing an empty collection to store the partial path-wise mappings (line~2, procedure \textsc{ConstructMap}). 
For each leaf node $u_i \in L_f$, it first computes the leaf-to-root path $P_i$ from $u_i$ to the root in $\widetilde{M}_f$, and the corresponding path $P_i'$ 
from its image $\phi(u_i)$ to the root in $\widetilde{M}_g$ (lines~4--5, procedure \textsc{ConstructMap}).
Using these two paths, the routine \textsc{ExtendMap} constructs a local mapping $\phi_i$ that extends the assignment $\phi(u_i)$ 
to all nodes along $P_i$ (line~6, procedure \textsc{ConstructMap}). 
These local maps are accumulated in the collection $\Phi$ (line~7, procedure \textsc{ConstructMap}). 
After all leaf-to-root paths are processed, the procedure invokes \textsc{IsWell-defined} to verify the global consistency of the collection $\Phi$ 
with respect to the least common ancestors (LCAs) of leaf pairs (line~9, procedure \textsc{ConstructMap}) where computation of LCAs will be discussed in Section~\ref{subsec: Is-epsilon-interleaved}. 
If the check succeeds, the partial maps are merged using \textsc{Map} to obtain a single continuous map $\varphi$ (line~10, procedure \textsc{ConstructMap}); 
otherwise, the procedure returns \texttt{Null} (line~13, procedure \textsc{ConstructMap}), indicating that no valid consistent mapping can be formed.


\begin{algorithmic}[1]
\label{proc:ConstructMap}
\Procedure{\textsc{ConstructMap}}{$\phi, \widetilde{M}_f, \widetilde{M}_g, L_f, \textsc{LCAs}$}
\State $\Phi\gets \emptyset$
    \For {each $u_i\in L_f$}
        \State Compute: $P_i \gets$ \textsc{NodeToRootPath}$(u_i, \widetilde{M}_f)$
        \State $P_i' \gets \textsc{NodeToRootPath}(\phi(u_i), \widetilde{M}_g)$
        \State $\phi_i \gets \textsc{ExtendMap}(u_i , \phi[u_i], P_i, P_i')$
        \State $\Phi.\textsc{AddToList}(\phi_i)$
    \EndFor
    \If{\textsc{IsWell-defined}$(\Phi, L_f,  \textsc{LCAs})$}
   \State $\varphi \gets \textsc{Map}(\Phi)$
   \State \Return{ $\varphi$ }
   \Else
   \State \Return{ \texttt{Null} }
    \EndIf
    \EndProcedure
\end{algorithmic}


\paragraph{Procedure: {\sc NodeToRootPath}.} 
The procedure \textsc{NodeToRootPath}$(u, \widetilde{M}_f)$ constructs the path from a node $u$ to the root of the merge tree $\widetilde{M}_f$. 
It initializes an empty list $P$ (lines~2--3, procedure \textsc{NodeToRootPath}) and iteratively adds each node along the upward traversal from $u$ to the root (lines~5--8, procedure \textsc{NodeToRootPath}). 
The resulting list $P$, containing all nodes on the path in order from leaf to root, is then returned (line~9, procedure \textsc{NodeToRootPath}).

\begin{algorithmic}[1]
\Procedure{\textsc{NodeToRootPath}}{$u, \widetilde{M}_f$}
    \State \% \textit{Initialize empty path list}
    \State $P \gets [\,]$ 
    \State $u' \gets u$
    \While{$u' \neq \texttt{Null}$}
        \State $P.\textsc{AddToList}(u')$
        \State $u' \gets \mathrm{parent}(u')$
    \EndWhile

    \State \Return $P$
\EndProcedure
\end{algorithmic}

\paragraph{Procedure: {\sc ExtendMap}.} 
The procedure \textsc{ExtendMap}$(x, y, P, P')$ extends a partial correspondence between two augmented merge trees along specified node-to-root paths. Given nodes $x \in \widetilde{M}_f$ and $y \in \widetilde{M}_g$ with $\g(y)=\f(x)+\varepsilon$, together with their respective paths $P$ and $P'$ from $x$ and $y$ to the corresponding roots, the procedure constructs a mapping $\psi$ that aligns each ancestor of $x$ with the corresponding ancestor of $y$. It begins by initializing the mapping with $\psi[x] = y$ (line~4, procedure \textsc{ExtendMap}) and then iteratively assigns each parent of $x$ in $P$ to the corresponding parent of $y$ in $P'$ while traversing upward along both paths (lines~7--11, procedure \textsc{ExtendMap}). 
As ensured by the augmentation step, there exists a one-to-one correspondence between the nodes of $P$ and $P'$. 
Consequently, the process proceeds until both paths reach their respective roots simultaneously, 
guaranteeing that each node $u \in P$ is mapped to a corresponding node in $P'$ whose $\g$-value satisfies $\g(\varphi(u)) = \f(u) + \varepsilon$.
\begin{algorithmic}[1]
\Procedure{\textsc{ExtendMap}}{$x, y, P, P'$}
    \State \% \textit{Initialize empty hash map}
    \State $\psi \gets \emptyset$ 
    \State $\psi[x] \gets y$
    \State $x' \gets P.\mathrm{parent}(x)$
    \State $y' \gets P'.\mathrm{parent}(y)$
    \While{$x' \neq \texttt{Null}~ \And ~y' \neq \texttt{Null}$}
        \State $\psi[x'] \gets y'$
        \State $x' \gets P.\mathrm{parent}(x')$
        \State $y' \gets P'.\mathrm{parent}(y')$
    \EndWhile
    \State \Return $\psi$
\EndProcedure
\end{algorithmic}


\paragraph{Procedure: {\sc IsWellDefined}.}  The procedure \textsc{IsWellDefined} verifies whether the collection of local mappings $\{\phi_i : P_i \rightarrow P_i'\}$ in $\Phi$, 
obtained from the \textsc{ExtendMap} procedure for all path pairs $(P_i, P_i')$, 
collectively defines a well-defined continuous map $\varphi : \widetilde{M}_f \rightarrow \widetilde{M}_g$.
 For each pair of distinct leaf nodes \( (u_i, u_j) \), it retrieves their least common ancestor \( v \) from the precomputed dictionary \textsc{LCAs}. Using this information, \textsc{IsWellDefined} determines whether \( \varphi \) is well-defined by checking that each pair of extensions in \( \Phi \) agrees at the corresponding least common ancestor. If the associated extensions 
\(\Phi[i]\) and \(\Phi[j]\) yield distinct images at  
\(v\), i.e., \(\Phi[i](v) \ne \Phi[j](v)\), then \(\varphi\) fails the 
well-definedness criterion of Theorem~\ref{theorem:continuity} (line 5, procedure \textsc{IsWellDefined}). Hence the
procedure returns \texttt{False} (line 6, procedure \textsc{IsWellDefined}).

 \begin{algorithmic}[1]
\Procedure{\textsc{IsWellDefined}}{$\Phi , L_f, \textsc{LCAs}$}

    \For {$u_i \in L_f$}
     
        \For {$u_j \in L_f ~\And~ u_j \neq u_i$}
      
                \State $v\gets \textsc{LCAs}[(u_i, u_j)]$

                \If{$\Phi[i] (v) \ne \Phi[j](v)$}
                       \State \Return{$\texttt{False}$}
                \EndIf
            \EndFor
        \EndFor
        \State \Return{$\texttt{True}$}
\EndProcedure
\end{algorithmic}

\paragraph{Procedure: {\sc Map}.} If the procedure \textsc{IsWellDefined} confirms the well-definedness of $\Phi$, 
then the procedure \textsc{Map} constructs the global mapping 
$\varphi=\bigcup_{i=1}^{\eta_f}\Phi[i]$. 
The procedure \textsc{Map} consolidates the collection of local mappings $\Phi$ into a single global map $\varphi$ 
by inserting each key--value pair $(u, \phi_i(u))$ for every node $u \in P_i$, 
where $P_i$ denotes the domain of the local map $\phi_i \in \Phi$, into $\varphi$ 
(lines~3--7, procedure~\textsc{Map}).
The detailed steps of \textsc{Map} are given below.
 \begin{algorithmic}[1]
\Procedure{Map}{$\Phi$}
    \State Initialize $\varphi \gets \emptyset$
    \For{each  $\phi_i \in \Phi$}
        \For{\textcolor{black}{each $u \in P_i$}}
            \State $\varphi[u] \gets \phi_i[u]$
        \EndFor
    \EndFor
    \State \Return $\varphi$
\EndProcedure
\end{algorithmic}

Once the map \(\varphi : \widetilde{M}_f \rightarrow \widetilde{M}_g\) has been 
constructed, the next step is to verify whether \(\varphi\) satisfies all the 
conditions required for it to qualify as an \(\varepsilon\)-good map. 
The verification of \(\varepsilon\)-goodness is performed in the following section.

\subsection{Verifying $\varepsilon$-Goodness of Map.} 
\label{subsubsec: Checking epsilon-goodness}
To establish that the constructed map 
$\varphi: \widetilde{M}_f \rightarrow \widetilde{M}_g$ 
is $\varepsilon$-good, we must verify that it satisfies all the properties of 
Definition~\ref{def:epsilon-goodmap}. 
The Range-Shift property holds by construction. 
The remaining task is to check the Ancestor-Shift and Ancestor-Closeness properties. 
These verifications are carried out in the procedure \textsc{Is-$\varepsilon$-Good} as 
described below.


\begin{algorithmic}[1]
\label{proc:IsEpsilonGood}
\Procedure{\textsc{Is-$\varepsilon$-Good}}{$\varphi, \widetilde{M}_f, \widetilde{M}_g, L_f, L_g, \textsc{LCAs},2\varepsilon\textsc{Pairs},\varepsilon$}

    \For {$u_i \in L_f$}

        \For {$u_j \in L_f ~\And~u_i\neq u_j$}
                \State $(u_r,u_s) \gets \textcolor{black}{2\varepsilon\textsc{Pairs}[(u_i,u_j)]}$

                \If{$(u_r,u_s) \ne \texttt{Null}~\And~ \varphi(u_r) = \varphi(u_s)$}
                       \State \Return{$\texttt{False}$}
                \EndIf
            \EndFor
        \EndFor
        \If{$\textsc{IsAncestorCloseness}(\varphi, \widetilde{M}_g,L_g, \varepsilon)$}
        \State \Return{$\texttt{True}$}
        \Else
        \State \Return{$\texttt{False}$}
    \EndIf
\EndProcedure
\end{algorithmic}

Here, the Ancestor-Shift property is verified by examining the precomputed $2\varepsilon$-pair $(u_r, u_s)$, which is stored in \textcolor{black}{the dictionary $2\varepsilon\textsc{Pairs}[(u_i,u_j)]$} corresponding to the pair of leaf nodes $(u_i, u_j)$ with $v = \LCA(u_i, u_j)$ (line~4, procedure \textsc{Is-$\varepsilon$-Good}). The procedure for computing the $2\varepsilon$-pair corresponding to each pair of leaf nodes will be discussed in Section~\ref{subsec: Is-epsilon-interleaved}. If the condition $\varphi(u_r) = \varphi(u_s)$ holds \textcolor{black}{(line~5, procedure \textsc{Is-$\varepsilon$-Good})}, then, by Theorem~\ref{theorem:Condition (b)}, the map $\varphi$ violates the Ancestor-Shift property, and the procedure returns \texttt{False} (line~6, procedure \textsc{Is-$\varepsilon$-Good}). Otherwise, the procedure \textsc{IsAncestorCloseness} verifies whether $\varphi$ satisfies the Ancestor-Closeness property (line~10, procedure \textsc{Is-$\varepsilon$-Good}). If this condition is met, $\varphi$ qualifies as an $\varepsilon$-good map, andprocedure \textsc{Is-$\varepsilon$-Good} returns \texttt{True} (line~11, procedure \textsc{Is-$\varepsilon$-Good}).

\paragraph{Procedure: {\sc IsAncestorCloseness}.}  The procedure \textsc{IsAncestorCloseness} verifies the \linebreak Ancestor-Closeness property of $\varphi$ (line~11, procedure \textsc{Is-$\varepsilon$-Good}). For each leaf node $w \in L_g \setminus \mathrm{Im}(\varphi)$, the procedure \textsc{NearestAncestor} identifies the nearest ancestor $w^a$ of $w$ that belongs to $\mathrm{Im}(\varphi)$ (line~3, procedure \textsc{IsAncestorCloseness}). If the difference in function values between $w$ and $w^a$ exceeds $2\varepsilon$, the procedure returns \texttt{False} (lines~4--5, procedure \textsc{IsAncestorCloseness}). Otherwise, $\varphi$ satisfies the Ancestor-Closeness property, as established in Lemma~\ref{lemma:ancestor-closeness}.

\begin{algorithmic}[1]
 \Procedure{\textsc{IsAncestorCloseness}}{$\varphi, \widetilde{M}_g,L_g, \varepsilon$}

\For{each $w \in L_g \setminus \mathrm{Im}(\varphi)$}
    \State $w^a \gets \textsc{NearestAncestor}(w, \mathrm{Im}(\varphi))$
    \If{$|\tilde{g}(w^a) - \tilde{g}(w)| > 2\varepsilon$}
        \State \Return{\texttt{False}}
    \EndIf
\EndFor
\State \Return{\texttt{True}}
\EndProcedure
\end{algorithmic}

\paragraph{Procedure: {\sc NearestAncestor}.}
The procedure \textsc{NearestAncestor} begins at any leaf node 
\( w \in L_g \setminus \mathrm{Im}(\varphi) \) and iteratively traverses upward through its ancestors (line 9, procedure \textsc{NearestAncestor}). When it encounters a node $w'$ that belongs to $\mathrm{Im}(\varphi)$, it assigns the nearest ancestor $w^a$ of $w$ to be $w'$ and terminates (lines 5--7, procedure \textsc{NearestAncestor}). 

\begin{algorithmic}[1]
\Procedure{\textsc{NearestAncestor}}{$ w, \mathrm{Im}(\varphi)$}
    \State $w^a \gets \texttt{Null}$
    \State $w' \gets w$
    \While{$w' \neq \texttt{Null}$}
        \If{$w' \in \mathrm{Im}(\varphi)$}
            \State $w^a \gets w'$
            \State \textbf{break}
            \Else
     \State $w' \gets \mathrm{parent}(w')$

          \EndIf
    \EndWhile
    \State \Return $w^a$
\EndProcedure
\end{algorithmic}

This completes the process of verifying whether a mapping \( \varphi: \widetilde{M}_f \rightarrow \widetilde{M}_g \) satisfies the criteria to be an \( \varepsilon \)-good map. If at least one such mapping \( \varphi \) is found to be \( \varepsilon \)-good, then $M_f$ and $M_g$ are $\varepsilon$-interleaved. In the next section, we combine all the steps in a single procedure to check whether $M_f$ and $M_g$ are $\varepsilon$-interleaved.

\subsection{Procedure: Is $\varepsilon$-Interleaved}
\label{subsec: Is-epsilon-interleaved}
In this section, we integrate the procedures \textsc{AugmentMergeTree}, \textsc{ConstructMap}, and \textsc{Is-$\varepsilon$-Good} into an algorithm, \textsc{Is-$\varepsilon$-Interleaved}, which determines whether the merge trees \( M_f \) and \( M_g \) are \( \varepsilon \)-interleaved. A formal description of the \textsc{Is-$\varepsilon$-Interleaved} algorithm is presented below.

\begin{algorithmic}[1]
\label{proc:IsEpsilonInterleaved}
\Procedure{\textsc{Is-$\varepsilon$-Interleaved}}{$M_f, M_g, \varepsilon$}
\State \%\textit{Extend the Merge Tree $M_g$}
\State $\{M_f, M_g\}\gets \textsc{ExtendMergeTree}(M_g, M_f, \varepsilon)$
\State \%\textit{Augmenting Merge Trees}


\State $\widetilde{M}_f \gets \textsc{AugmentMergeTree}(M_f, M_g, \varepsilon)$

\State $\widetilde{M}_g \gets \textsc{AugmentMergeTree}(M_g, M_f, -\varepsilon)$
\State $L_f\gets \Mf.\textsc{ListOfLeafNodes}()$
\State $L_g\gets \Mg.\textsc{ListOfLeafNodes}()$
\State $u_{f} \gets  L_f.\mathrm{min}()$
\State $u_{g} \gets L_g.\mathrm{min}()$
\If{$\g(u_{g}) -\f(u_{f}) > \varepsilon$}
\State \Return{\texttt{False}}
\EndIf

  \State \% \textit{Initialize an empty HashMap}
  \State $\textsc{LCAs} \gets \emptyset$ 
  \State $2\varepsilon\textsc{Pairs} \gets \emptyset$ 
  \For{each leaf node $u_i \in L_f$}
    \For{each leaf node $u_j \in L_f$}
    \State  \% \textit{Avoid duplicate and self-pairs}
        \If{$i < j$} 
            \State $v \gets \textsc{FindLCA}(\Mf, u_i, u_j)$
            \State $\text{LCAs}[(u_i, u_j)] \gets v$
            \State $(u_r, u_s) \gets \textsc{Find2$\varepsilon$Pair}(\Mf, u_i, u_j, v)$
            \State $2\varepsilon\textsc{Pairs}[(u_i, u_j)]\gets (u_r, u_s)$
        \EndIf
    \EndFor
\EndFor

\State \% \textit{Initialize an empty HashMap}
\State $\Gamma \gets \emptyset$ 
\For {each $u_i\in L_f$}
    \State $\Gamma[u_i]\gets \textsc{TargetNodes}(\Mg, \f(u_i) + \varepsilon)$
\EndFor
\State \%\textit{Constructing Maps}
   \For{ each $(\hat{u}_1,\hat{u}_2,\ldots, \hat{u}_{\eta_f}) \in \Gamma[u_1] \times \Gamma[u_2] \times \dots\times \Gamma[u_{\eta_f}]$}
   \State \% \textit{Initialize an empty HashMap}
   \State $\phi\gets \emptyset$ 
     \For{$u_i \in L_f$}
     \State $\phi[u_i] \gets \hat{u}_i$
     \EndFor
     \State $\varphi \gets \textsc{ConstructMap}(\phi, \widetilde{M}_f, \widetilde{M}_g, L_f, \textsc{LCAs})$
     \State \%Verifying $\varepsilon$-Goodness
     \If{ $\varphi \ne \texttt{Null}$  $\And$   \textsc{Is-$\varepsilon$-Good}($\varphi, \widetilde{M}_f, \widetilde{M}_g, L_f, L_g, \textsc{LCAs},2\varepsilon\textsc{Pairs},\varepsilon$)}
     \State \Return{\texttt{True}}
   
     \EndIf
   \EndFor
   \State \Return{ \texttt{False}}
\EndProcedure
\end{algorithmic}

In procedure \textsc{Is-$\varepsilon$-Interleaved}, we first extend the merge trees $M_f$ and $M_g$ to \textcolor{black}{construct} an $\varepsilon$-good map between $M_f$ and $M_g$ (lines 2--3, procedure \textsc{Is-$\varepsilon$-Interleaved}) and then augment the merge trees $M_f$ and $M_g$ to obtain the augmented trees $\widetilde{M}_f$ and $\widetilde{M}_g$ (lines 4--6, procedure \textsc{Is-$\varepsilon$-Interleaved}). Next, the procedure \textsc{ListOfLeafNodes} extracts all leaf nodes from each tree, while the procedure $\mathrm{min}()$ identifies the leaf nodes in $\widetilde{M}_f$ and $\widetilde{M}_g$ with the minimum function values (lines 7--10, procedure \textsc{Is-$\varepsilon$-Interleaved}). If the difference in function values, $\tilde{g}(u_g) - \tilde{f}(u_f)$, where $u_g$ and $u_f$ are the leaf nodes with the smallest function values in $\widetilde{M}_g$ and $\widetilde{M}_f$, respectively, exceeds $\varepsilon$, then constructing an $\varepsilon$-good map between $M_f$ and $M_g$ is impossible, and the procedure immediately returns \texttt{False} (lines 11--13, procedure \textsc{Is-$\varepsilon$-Interleaved}). For each pair of leaf nodes $(u_i, u_j)$, the procedure \textsc{FindLCA} identifies their least common ancestor $v$ and stores the result in the dictionary \textsc{LCAs} (lines 21--22, procedure \textsc{Is-$\varepsilon$-Interleaved}). Additionally, the procedure \textsc{Find2$\varepsilon$Pair} determines $2\varepsilon$-pair $(u_r, u_s)$ corresponding to \textcolor{black}{the pair} of leaf nodes $(u_i, u_j)$ and stores in a dictionary \textcolor{black}{\textsc{$2\varepsilon$Pairs}} (lines 23--24, procedure \textsc{Is-$\varepsilon$-Interleaved}). For each leaf node in $\widetilde{M}_f$, the procedure \textsc{TargetNodes} computes the corresponding list of target nodes in $\widetilde{M}_g$ by iterating through all the nodes of the tree with specific function value (lines 30--32, procedure \textsc{Is-$\varepsilon$-Interleaved}). Using these lists of target nodes, an initial map $\phi$ is constructed on the set of leaf nodes $L_f$ (lines 37--39, procedure \textsc{Is-$\varepsilon$-Interleaved}). The procedure \textsc{ConstructMap} then extends $\phi$ to a map $\varphi : \widetilde{M}_f \rightarrow \widetilde{M}_g$ (line 40, procedure \textsc{Is-$\varepsilon$-Interleaved}). If $\varphi$ is not well-defined, it returns \texttt{Null}; otherwise, the procedure \textsc{Is-$\varepsilon$-Good} verifies whether $\varphi$ satisfies the $\varepsilon$-goodness conditions (line 42, procedure \textsc{Is-$\varepsilon$-Interleaved}). If at least one $\varepsilon$-good map exists, \textsc{Is-$\varepsilon$-Good} returns \texttt{True} (line 43, procedure \textsc{Is-$\varepsilon$-Interleaved}); otherwise, it returns \texttt{False} (line 46, procedure \textsc{Is-$\varepsilon$-Interleaved}). The individual procedures invoked in \textsc{Is-$\varepsilon$-Interleaved} are elaborated below.

\paragraph{Procedure: {\sc FindLCA}.}
The procedure \textsc{FindLCA} computes the least common ancestor (LCA) of two nodes $u_i$ and $u_j$ in the augmented merge tree $\widetilde{M}_f$. It first obtains the root-to-node paths \( \bar{P_i} \) and \( \bar{P_j} \) using the procedure \textsc{PathFromRoot} (lines~2--3, procedure \textsc{FindLCA}), which is analogous to the procedure \textsc{NodeToRootPath}  in Section~\ref{subsec:constructing-maps}. The algorithm then iterates simultaneously through both paths, comparing corresponding nodes $(x, y)$, with $\tilde{f}(x) = \tilde{f}(y)$, in order from root to leaf (line~5, procedure \textsc{FindLCA}). 
Whenever the function values satisfy $\tilde{f}(x) = \tilde{f}(y)$ and the nodes are identical ($x = y$), the current node is recorded as the LCA (lines~6--7, procedure \textsc{FindLCA}). 
The iteration continues until a mismatch occurs, at which point the loop terminates (lines~8--9, procedure \textsc{FindLCA}). 
Finally, the procedure returns the node stored in `$\mathrm{lca}$', representing the least common ancestor of $u_i$ and $u_j$ (line~12, procedure \textsc{FindLCA}).



\begin{algorithmic}[1]
\Procedure{\textsc{FindLCA}}{$\widetilde{M}_f, u_i,u_j$}
    \State $\bar{P_i} \gets \textsc{PathFromRoot}({u_i})$
\State $\bar{P_j} \gets \textsc{PathFromRoot}({u_j})$
\State $\text{lca} \gets \texttt{Null}$
\For{each pair of nodes $(x, y)$ from $\bar{P_i}$ and $\bar{P_j}$, respectively, with $\f(x)=\f(y)$, in order from root to leaf}
    \If{$x = y$}
        \State $\text{lca} \gets x$
    \Else
        \State \textbf{break}
    \EndIf
\EndFor
\State \Return $\text{lca}$
\EndProcedure
\end{algorithmic}

\paragraph{Procedure: {\sc Find$2\varepsilon$Pair}.}
The procedure \textsc{Find$2\varepsilon$Pair} identifies the first pair of descendant nodes of the least common ancestor \( v = \LCA(u_i, u_j) \) whose function values differ from \( \f(v) \) by more than \( 2\varepsilon \). It traverses the root-to-node paths \( \bar{P_i} \) and \( \bar{P_j} \) corresponding to \( u_i \) and \( u_j \), and \textcolor{black}{searches for} the first pair \((u_i', u_j')\) with equal $\f$-value and satisfying  $\f(v) - \f(u_i')>2\varepsilon$ (lines 5--6, procedure \textsc{Find$2\varepsilon$Pair}). Once the first such pair is encountered, it is recorded and returned (line 7, procedure \textsc{Find$2\varepsilon$Pair}). Otherwise,  the procedure returns \texttt{Null} when $u_i'$ or $u_j'$, or both do not exist.

\begin{algorithmic}[1]
\Procedure{\textsc{Find$2\varepsilon$Pair}}{$\widetilde{M}_f, u_i,u_j,v$}
  
    \State $\bar{P_i} \gets \textsc{PathFromRoot}({u_i})$
\State $\bar{P_j} \gets \textsc{PathFromRoot}({u_j})$
\State \textcolor{black}{$(u_r,u_s) \gets \texttt{Null}$}
\For{each distinct pair \textcolor{black}{$(u_i', u_j')$}  of nodes from $\bar{P_i}$ and $\bar{P_j}$ with same $\f$-value  }
    \If{$\f(v) - \f(u_i') > 2\varepsilon$}
        \State $ (u_r,u_s)\gets (u_i',u_j')$
        \State \textbf{break}
    \Else
        \State $u_i' \gets P_i.\mathrm{child}(u_i')$
        \State $u_i' \gets P_j.\mathrm{child}(u_j')$
    \EndIf
\EndFor

\State \Return \textcolor{black}{$(u_r,u_s)$}
    
\EndProcedure
\end{algorithmic}

\color{black}

\paragraph{Procedure: {\sc TargetNodes}.}
 The procedure \textsc{TargetNodes} begins by initializing an empty list \texttt{Result} to store the matching nodes (line 2, procedure \textsc{TargetNodes}).
 It then iterates through all the nodes in $\Mg$ and identifies all nodes whose function value equals a specific value \(c\) (lines 4--8, \textsc{TargetNodes}). 
Finally, the list \texttt{Result}, containing all nodes whose function values equal to $c$, is returned (line 9, procedure \textsc{TargetNodes}).

\begin{algorithmic}[1]
\Procedure{\textsc{TargetNodes}}{$\Mg$, $c$}
   \State $\texttt{Result} \gets \emptyset$
\State $V_g \gets \Mg.\textsc{GetNodes}()$
\For{$v \in V_f$}
   \If{$|\tilde{g}(v) - c| = 0$}
        \State $\texttt{Result}.\mathrm{append}(v)$
    \EndIf
 \EndFor   
\State \Return $\texttt{Result}$
\EndProcedure
\end{algorithmic}
\color{black}

The procedure \textsc{Is-$\varepsilon$-Interleaved} is invoked from the main Algorithm~\ref{alg:main} (line~9) 
to verify whether the given merge trees $M_f$ and $M_g$ are $\varepsilon$-interleaved. 
Upon completion, Algorithm~\ref{alg:main} returns the interleaving distance between $M_f$ and $M_g$. 
However, for practical applications, Algorithm~\ref{alg:main} can be computationally expensive, 
as it may need to consider up to $\eta_g^{\eta_f}$ possible mappings, 
where $\eta_f$ and $\eta_g$ denote the numbers of leaf nodes in $M_f$ and $M_g$, respectively. 
In the next section, we introduce a refined and more efficient approach to mitigate this computational cost.

\subsection{Refinement of the Algorithm}
\label{Subsec: Refinement}
The procedure \textsc{RefinedTargetNodes} 
computes the refined set of valid target nodes $L^{u_k}$ in $\widetilde{M}_g$ for a given leaf \textcolor{black}{node} $u_k \in L_f$.

\begin{algorithmic}[1]
 \Procedure{\textsc{RefinedTargetNodes}}{$u_k,\widetilde{M}_f,\widetilde{M}_g,L_f,\textsc{LCAs}, \textsc{$2\varepsilon$Pairs}, \varepsilon$}
\State \% \textit{Initialize an empty HashMap}
\State $\mathcal{P} \gets \emptyset$

\State $\Gamma \gets \emptyset$

\For{each $u_i \in L_f$}
    \State $\mathcal{P}[u_i] \gets \textsc{NodeToRootPath}(u_i, \widetilde{M}_f)$
    \State $\Gamma[u_i] \gets \textsc{TargetNodes}(\Mg, \f(u_i) + \varepsilon)$  
\EndFor
    \State \% \textit{Initialize an empty list for refined target nodes}
    \State $L^{u_k} \gets \emptyset$
 
    \State $P_k \gets \mathcal{P}[u_k]$
    
    \For{each $\hat{u}_i \in \Gamma[u_k]$}
    \State \% \textit{Checks if $\hat{u}_i$ is a valid target node for $u_k$}
       \State $\texttt{IsValidTarget} \gets \texttt{True}$ 
    \State $P_i' \gets \textsc{NodeToRootPath}(\hat{u}_i,\widetilde{M}_g)$
        \State $\phi_k \gets \textsc{ExtendMap}(u_k,\hat{u}_i, P_k, P_i')$

        \For{each $u_{\ell} \in L_f$ $\And~u_{\ell} \ne u_k$ }
        \State \% \textit{Checks if there is a map $\phi_\ell$ compatible to  $\phi_k$}
            \State $\texttt{HasCompatibleMap} \gets \texttt{False}$
            
            \State $P_{\ell} \gets \mathcal{P}[u_{\ell}]$
            
            \State $v \gets \textsc{LCAs}[(u_k, u_{\ell})]$
            \State $(u_r, u_s)\gets \textsc{$2\varepsilon$Pairs}[(u_k,u_{\ell})]$

            \For{each $\hat{u}_\ell \in \Gamma[u_{\ell}]$}
            
              \State $P_{\ell}' \gets \textsc{NodeToRootPath}(\hat{u}_\ell,\widetilde{M}_g)$
        \State $\phi_\ell \gets \textsc{ExtendMap}(u_{\ell},\hat{u}_\ell, P_{\ell}, P_{\ell}')$

              \color{black}  \If{$\phi_k[v] = \phi_\ell[v]$ $\And$ $(u_r, u_s)=\texttt{Null}$}
                    \State $\texttt{HasCompatibleMap} \gets \texttt{True}$
                        \State \textbf{break}
                \EndIf
                \If{$\phi_k[v] = \phi_\ell[v]$ $\And$ $\phi_k[u_r] \ne \phi_\ell[u_s]$}
                    
                        \State $\texttt{HasCompatibleMap} \gets \texttt{True}$
                        \State \textbf{break}
                \EndIf
 \color{black}\EndFor

            \If{\texttt{HasCompatibleMap} = \texttt{False}}
                \State $\texttt{IsValidTarget} \gets \texttt{False}$
                \State \textbf{break}
            \EndIf
        \EndFor

        \If{\texttt{IsValidTarget}}
            \State $L^{u_k}.\textsc{AddToList}(\hat{u}_i)$
        \EndIf
    \EndFor
\State \Return{$L^{u_k}$}
\EndProcedure
\end{algorithmic}
It first precomputes and stores each leaf’s node-to-root path $\mathcal{P}[u_i]$ and its candidate target nodes $\Gamma[u_i]$ (lines~3--8, procedure \textsc{RefinedTargetNodes}). 
For each candidate $\hat{u}_i \in \Gamma[u_k]$, the procedure builds the local extension $\phi_k$ from $u_k$ along its node-to-root path $P_k$ \textcolor{black}{to the node-to-root path $P_i'$ from $\hat{u}_i$ in $\Mg$} (lines~11--16, procedure \textsc{RefinedTargetNodes}). 
It then checks compatibility of $\phi_k$ with each mapping $\phi_\ell$ corresponding to every other leaf $u_{\ell}$ by (i) retrieving the precomputed path $P_{\ell}$, 
the LCA $v=\textsc{LCAs}[(u_k,u_{\ell})]$, and the $2\varepsilon$-pair $(u_r,u_s)=\textsc{$2\varepsilon$Pairs}[(u_k,u_{\ell})]$ \textcolor{black}{as described in} Theorem~\ref{thm:ancestor-shift-condition}  (lines~20--22, procedure \textsc{RefinedTargetNodes}), 
and (ii) searching for at least one candidate $\hat{u}_{\ell} \in \Gamma[u_{\ell}]$ whose extension $\phi_{\ell}$ agrees with $\phi_k$ 
at the LCA (i.e., $\phi_k[v]=\phi_{\ell}[v]$) while distinguishing the images of $2\varepsilon$-pair (i.e., $\phi_k[u_r] \ne \phi_{\ell}[u_s]$) \textcolor{black}{(lines~23--34, procedure \textsc{RefinedTargetNodes})}. 
\textcolor{black}{Note that if $(u_r, u_s)=\texttt{Null}$ and $\phi_k[v]=\phi_{\ell}[v]$, then $\phi_\ell$ is compatible (lines 26--27, procedure \textsc{RefinedTargetNodes}).}
If no such compatible $\phi_{\ell}$ exists for some $u_{\ell}$, then the candidate $\hat{u}_i$ is rejected (lines~35--38, procedure \textsc{RefinedTargetNodes}); 
otherwise, it is accepted and added to $L^{u_k}$ (lines~40--42, procedure \textsc{RefinedTargetNodes}). 
After testing all candidates, the procedure returns the refined list $L^{u_k}$ (line~44, procedure \textsc{RefinedTargetNodes}).


If the maximum size of the refined list \( \kappa := \max_{u_i \in L_f} |L^{u_i}| \leq \eta_g \), then the number of mappings considered reduces from \( \eta_g^{\eta_f} \) to \( \kappa^{\eta_f} \). If \( L^{u_k} = \emptyset \) for \textcolor{black}{some} \( u_k \), no \( \varepsilon \)-good map exists between $M_f$ and $M_g$. Hence \( M_f \) and \( M_g \) are not \( \varepsilon \)-interleaved for the selected $\varepsilon$; the algorithm then increases \( \varepsilon \) in the binary search. Having introduced this refinement strategy, we next establish the correctness of the algorithm.
\subsection{Correctness}

 In this section, we provide formal justification for the validity of each step of the algorithm and proves that the algorithm indeed computes the exact interleaving distance $\varepsilon^* = d_I(M_f,M_g)$ between merge trees $M_f$ and $M_g$. The theoretical results introduced earlier forms the foundation of the following Theorem.

\begin{theorem}[\textbf{Correctness}]
Given two merge trees \( M_f \) and \( M_g \), the Algorithm~\ref{alg:main}   computes the exact interleaving distance \( \varepsilon^* = d_I(M_f, M_g) \) between $M_f$ and $M_g$. 
\end{theorem}

\begin{proof}
We have a finite sorted set \( \Pi \)  of candidate values for \( \varepsilon^* \)  as in Lemma \ref{lemma:candidate values}.  By Lemma~\ref{lemma: monotonicity}, if \( M_f \) and \( M_g \) are \( \varepsilon \)-interleaved, then they are also \( \varepsilon' \)-interleaved for any \( \varepsilon' \geq \varepsilon \), which ensures the correctness of binary search algorithm. At each step of the binary search, the procedure \textsc{Is-$\varepsilon$-Interleaved} verifies whether an \( \varepsilon \)-good map exists between $M_f$ and $M_g$. It does so by enumerating all possible \( \eta_g^{\eta_f} \) number of mappings \( \phi: L_f \rightarrow M_g \) for the leaf nodes in $L_f$. Then each map is extended to the entire augmented merge tree $\Mf$ using the path-based construction by ensuring the continuity using Theorem~\ref{theorem:continuity}. This extension $\varphi:\widetilde{M}_f \rightarrow \widetilde{M}_g$ satisfies the Range-Shift property by construction and we check Ancestor-Shift and Ancestor-Closeness properties using  Theorem~\ref{theorem:Condition (b)} and Lemma~\ref{lemma:ancestor-closeness} respectively. If \( \varphi \) satisfies all the properties for being \( \varepsilon \)-good, then the algorithm concludes that \( M_f \) and \( M_g \) are \( \varepsilon \)-interleaved. After completing all iterations of the binary search, the smallest such \( \varepsilon \in \Pi \) for which \( M_f \) and \( M_g \) are \( \varepsilon \)-interleaved is returned, denoted by \( \varepsilon^* \). Since all candidate values are considered, and for each candidate \( \varepsilon \) the algorithm iterates through all possible mappings to verify whether the properties of Definition~\ref{def:epsilon-goodmap} are satisfied, the algorithm is guaranteed to return the exact interleaving distance \( \varepsilon^* \) between \( M_f \) and \( M_g \).
\end{proof}

Once correctness has been established, we analyze the computational complexity \textcolor{black}{for} the algorithm. We examine the time required of each component of our algorithm in the following section.

\section{Complexity Analysis}
\label{sec:complexity}
Let $n$ denote the total number of nodes across the merge trees $M_f$ and $M_g$, with $\eta_f$ and $\eta_g$ their respective numbers of leaf nodes. In this setting, we construct maps 
either from $L_f$ to $\widetilde{M}_g$ or from $L_g$ to $\widetilde{M}_f$, 
depending on which yields the smaller count.  
For example, if $\eta_f = 10$ and $\eta_g = 100 = 10^2$, then the number of 
maps from $L_g$ to $\widetilde{M}_f$ is $\eta_f^{\eta_g} = 10^{100}$, while 
the number of maps from $L_f$ to $\widetilde{M}_g$ is 
$\eta_g^{\eta_f} = 10^{20}$, which is substantially smaller. 
In this case, we reduce the complexity by considering maps from $L_f$ to 
$\Mg$. Our algorithm consists of 
several key procedures, each contributing to the total runtime. The time complexity of each step in Section~\ref{Sec:algorithm} is described as follows.

\subsection{Complexity of the Procedure: {\sc GenerateCandidateValues }} The procedure \textsc{GenerateCandidateValues} computes the sorted list \( \Pi \) of candidate values for the interleaving distance. \textcolor{black}{Lines~13--23} calculate the function value differences between the nodes within individual merge trees, while \textcolor{black}{lines~7--12} handle differences across the two trees. Since each part yields \( O(n^2) \) combinations, the size of \( \Pi \) is \( O(n^2) \). Sorting \( \Pi \) requires \( O(n^2 \log n) \) time (\textcolor{black}{line~25}, procedure \textsc{GenerateCandidateValues}). Hence, \textsc{GenerateCandidateValues} runs in \( O(n^2 \log n) \) time.

\subsection{Complexity of the Procedure: {\sc ExtendMergeTree}}
\textcolor{black}{The procedure \textsc{ExtendMergeTree} runs in $O(1)$ time, as it only requires creating a new root node in either $M_f$ or $M_g$, 
depending on the function values of their original roots.
}  

\subsection{Complexity of the Procedure: {\sc AugmentMergeTree}}
 The procedure \textsc{AugmentMergeTree} iterates through all the nodes and all the edges of the merge tree $M_f$ or $M_g$, \textcolor{black}{and} it requires \( O(n^2) \) time (lines~7--8 and lines 19--20, procedure \textsc{AugmentMergeTree}).   
Hence, the overall time complexity of this procedure is $O(n^2)$.


\subsection{Complexity of the Procedure : {\sc ConstructMap}} 
\textcolor{black}{The procedure {\sc ConstructMap} first constructs the paths using \textsc{NodeToRootPath} which takes \( O(n) \) time for each path (lines~4--5, procedure \textsc{ConstructMap}), resulting in total time of  $O(n\eta_f)$ for all leaf-to-root paths in $\Mf$. 
Given a map \( \phi : L_f \rightarrow \widetilde{M}_g \), the procedure \textsc{ExtendMap} extends \( \phi \) along all leaf-to-root paths \( P_i \) from the leaf nodes \( u_i \in \widetilde{M}_f \), for \( i = 1, 2, \ldots, \eta_f \), by traversing each node on \( P_i \), which also requires \( O(n\eta_f) \) time (line~6, procedure~\textsc{ConstructMap}).  
For each of the \( \eta_f^2 \) pairs of maps \( (\phi_i, \phi_j) \), checking whether they agree on \( v = \textsc{LCAs}[(u_i, u_j)] \) takes \( O(1) \) time per pair (line~5, procedure~\textsc{IsWellDefined}), provided the dictionary \textsc{LCAs} has been constructed, resulting in \( O(\eta_f^2) \) total time. 
The procedure \textsc{Map} iterates through the nodes on all leaf-to-root paths, requiring \( O(n\eta_f) \) time. 
Therefore, constructing a single mapping \( \varphi : \widetilde{M}_f \rightarrow \widetilde{M}_g \) takes \( O(n\eta_f + \eta_f^2)\) total time.}

\subsection{Complexity of the Procedure: {\sc Is-$\varepsilon$-Good}}
The procedure \textsc{Is-$\varepsilon$-Good} verifies the Ancestor-Shift property 
by checking whether $\varphi(u_r) = \varphi(u_s)$ for each $2\varepsilon$-pair 
$(u_r, u_s) = \textsc{2$\varepsilon$Pairs}[(u_i, u_j)]$ corresponding to the leaf node pair $(u_i, u_j)$. 
This operation requires $O(1)$ time per pair (lines~4--5, procedure~\textsc{Is-$\varepsilon$-Good}), 
assuming that the dictionary \textsc{2$\varepsilon$Pairs} is precomputed. 
Consequently, verifying the Ancestor-Shift property for all pairs of leaf nodes takes $O(\eta_f^2)$ time in total.  
For the Ancestor-Closeness property, identifying the nearest ancestor 
$w^a \in \mathrm{Im}(\varphi)$ of a leaf node $w \in L_g \setminus \mathrm{Im}(\varphi)$ 
requires $O(n)$ time for a leaf-to-root search (line~3, procedure~\textsc{IsAncestorCloseness}). 
Therefore, checking this condition for all leaf nodes costs $O(n\eta_g)$ time. 
Hence, the overall time complexity of the procedure \textsc{Is-$\varepsilon$-Good} is $O(n\eta_g + \eta_f^2)$.

\subsection{Complexity of the Procedure: {\sc Is-$\varepsilon$-Interleaved}}  
The procedure \textsc{Is-$\varepsilon$-Interleaved} integrates four main components: 
\textsc{ExtendMergeTree}, \textsc{AugmentMergeTree}, \textsc{ConstructMap}, 
and \textsc{Is-$\varepsilon$-Good}. 
The time complexities of these procedures are as follows: 
\textsc{ExtendMergeTree} runs in $O(1)$ time, as it only involves creating a new root node in either $M_f$ or $M_g$; 
\textsc{AugmentMergeTree} requires $O(n^2)$ time; 
\textsc{ConstructMap} takes $O(n\eta_f + \eta_f^2)$ time; 
and \textsc{Is-$\varepsilon$-Good} has a time complexity of $O(n\eta_g + \eta_f^2)$. 

The procedure \textsc{ListOfLeafNodes} extracts all leaf nodes by traversing the entire tree, 
which requires at most $O(n)$ time \textcolor{black}{(lines~7--8, procedure~\textsc{Is-$\varepsilon$-Interleaved})}. 
Computing the node of $M_f$ or $M_g$ with minimum function value also takes $O(n)$ time (lines 9--10, procedure~\textsc{Is-$\varepsilon$-Interleaved}) . 
Therefore, lines~2--10 of \textsc{Is-$\varepsilon$-Interleaved} together contribute 
$O(1) + O(n^2) + O(n) + O(n) = O(n^2)$ time.

Next, the procedure \textsc{FindLCA} computes the least common ancestor ($\LCA$) of any pair of leaf nodes 
in $O(n)$ time in the worst case \textcolor{black}{(line~21, procedure~\textsc{Is-$\varepsilon$-Interleaved})}. 
Since there are $O(\eta_f^2)$ pairs of leaf nodes, computing all LCAs in $\widetilde{M}_f$ requires $O(n\eta_f^2)$ time. 
For each leaf pair $(u_i, u_j)$ with $v = \mathrm{LCA}(u_i, u_j)$, 
the procedure \textsc{Find$2\varepsilon$Pair}$(u_i, u_j)$ identifies the first pair of nodes $(u_r, u_s)$ 
for which the function value difference between $u_r$ and $v$ exceeds $2\varepsilon$, 
by traversing downward from $v$ along the leaf-to-root paths of $u_i$ and $u_j$ 
\textcolor{black}{(line~23, procedure~\textsc{Is-$\varepsilon$-Interleaved})}. 
This requires $O(n)$ time per pair, resulting in a total of $O(n\eta_f^2)$ time for all leaf pairs. 
Hence, lines~14--27 together take $O(n\eta_f^2)$ time.  

The procedure \textsc{TargetNodes} traverses all nodes to identify those lying at a specific level, 
requiring $O(n)$ time per call \textcolor{black}{(line~31, procedure~\textsc{Is-$\varepsilon$-Interleaved})}. 
For all $\eta_f$ leaf nodes, this step therefore costs $O(n\eta_f)$ time 
\textcolor{black}{(lines~30--32, procedure~\textsc{Is-$\varepsilon$-Interleaved})}.  

There are $\eta_g^{\eta_f}$ total possible mappings 
$\phi : L_f \to \widetilde{M}_g$ \textcolor{black}{(line~34, procedure~\textsc{Is-$\varepsilon$-Interleaved})}, 
as discussed in Step~3 of Section~\ref{Sec:algorithm}. 
Since constructing a single such mapping requires $O(\eta_f)$ time 
\textcolor{black}{(lines~37--39, procedure~\textsc{Is-$\varepsilon$-Interleaved})}, 
and verifying each map involves additional costs from \textsc{ConstructMap} and \textsc{Is-$\varepsilon$-Good}, 
the total complexity for lines~34--45 is  
\[
O\!\left(\eta_g^{\eta_f}(\eta_f + n\eta_f + \eta_f^2 + n\eta_g + \eta_f^2)\right)
= O\!\left(\eta_g^{\eta_f}(n\eta_f + n\eta_g)\right),
\]
since $\eta_f \le n$.  

Combining all steps, the total time complexity of \textsc{Is-$\varepsilon$-Interleaved} is  
\[
O(n^2) + O(n\eta_f^2) + O(n\eta_f) + O\!\left(\eta_g^{\eta_f}(n\eta_f + n\eta_g)\right).
\]
For $\eta_g \ge 2$, we have $n\eta_f^2 < \eta_g^{\eta_f}(n\eta_f + n\eta_g)$, 
implying $n\eta_f^2 = O\!\left(\eta_g^{\eta_f}(n\eta_f + n\eta_g)\right)$.  
Therefore, the overall time complexity of the procedure \textsc{Is-$\varepsilon$-Interleaved}, 
for each chosen value of $\varepsilon$, is  
\[
O\!\left(n^2 + \eta_g^{\eta_f}(n\eta_f + n\eta_g)\right).
\]
 
\paragraph{\textbf{Complexity of Algorithm~\ref{alg:main}:}} The procedure \textsc{GenerateCandidateValues} takes \( O(n^2 \log n) \) time (line~3, Algorithm~\ref{alg:main}).  
Line~9 of Algorithm~\ref{alg:main} takes 
\( O\!\left(n^2 + \eta_g^{\eta_f}(n\eta_f + n\eta_g) \right) \) time 
for each chosen value of \( \varepsilon \).  
Performing a binary search over \( O(n^2) \) candidate values requires \( O(\log n) \) time 
(lines~6--15, Algorithm~\ref{alg:main}).  
Therefore, adding all,  the overall time complexity of Algorithm~\ref{alg:main} is $O\!\left(n^2 \log n + \big(n^2 + \eta_g^{\eta_f}(n\eta_f + n\eta_g)  \big)\log n\right)
$, i.e.,
$O\!\left(n^2 \log n \;+\; \eta_g^{\eta_f}(\eta_f + \eta_g) \cdot n \log n\right)$.

\color{black}
\subsection{Complexity of the Procedure: {\sc  RefinedTargetNodes}}

The procedure {\sc  RefinedTargetNodes} refines the set of target nodes for a given leaf node $u_k \in \Mf$. It begins by constructing a path $\mathcal{P}[u]$ from leaf node $u$ to the root of $\Mf$ and computing the set $\Gamma[u]$ of target nodes for $u$, which requires $O(n)$ time (lines~3--8, procedure \textsc{RefinedTargetNodes}). Repeating this for all leaf nodes in $\Mf$, it requires $O(n\eta_f)$. For each target node $\hat{u}_i \in \Gamma[u_k]$, we then verify its validity using the procedure $\texttt{IsValidTarget}$ (line 14, procedure \textsc{RefinedTargetNodes}). For every other leaf node $u_{\ell} \ne u_k$ in $\Mf$, we check whether there exists a target node $\hat{u}_\ell \in \Gamma[u_\ell]$ such that the path pair $(P_i', P_{\ell}')$ satisfies both the well-definedness and Ancestor-Shift property, where $P_{\ell}'$ denotes \textcolor{black}{the} path from $\hat{u}_\ell$ to the root of $\Mg$. Since this requires iterating over all  $u_{\ell} \in L_f$ and their corresponding $\eta_g$ possible target nodes in $\Gamma[u_\ell]$, it cost $O(n\eta_f\eta_g)$ time (lines 17--34, procedure \textsc{RefinedTargetNodes}). As we iterate over all $\eta_g$ target nodes of $u_k$ (line 12, procedure \textsc{RefinedTargetNodes}), the total time complexity of the procedure \textsc{RefinedTargetNodes} is $O(n\eta_f\eta_g^2)$.

\paragraph{\textbf{Complexity of Algorithm~\ref{alg:main} (After Refinement):}} Since the merge tree $\Mf$ has $\eta_f$ leaf nodes, computing the refined target nodes for all leaf nodes  requires $O(n\eta_f^2\eta_g^2)$ time, as the procedure \textsc{TargetNodes} (line~31, \textsc{Is-$\varepsilon$-Interleaved}) is replaced by \textsc{RefinedTargetNodes}. Let $\kappa$ denote the maximum size of $L^{u_i}$ for $i \in {1,2,\dots,\eta_f}$. Consequently, the procedure \textsc{Is-$\varepsilon$-Interleaved} takes \textcolor{black}{$O(n^2 + n\eta_f^2\eta_g^2+ \kappa^{\eta_f}( \eta_f + \eta_g)\cdot n)$} for each chosen value of $\varepsilon$ where the term $\kappa^{\eta_f}$ represents the total number of maps to be explored.  
Therefore, the overall time complexity of Algorithm~\ref{alg:main} after refinement is  
\[
O( n^2\log n +n\eta_f^2\eta_g^2 + \kappa^{\eta_f}  (\eta_f+\eta_g) \cdot n\log n).
\]
\color{black}
\color{black}
While both the degree bound parameter \( \tau \) and the numbers of leaf nodes  \( \eta_f \) and $\eta_g$ are both structural parameters of a merge tree, they are fundamentally incomparable. In the following section, we discuss and analyze the roles of these two parameters, highlighting their differences and implications on algorithmic complexity.

\section{Discussion}
\label{sec:discussion}
It is important to note that the parameter \( \tau \), introduced in \cite{touli2022fpt}, depends on the choice of \( \varepsilon \), where \( \varepsilon \) denotes the candidate value for the interleaving distance between the merge trees \( M_f \) and \( M_g \).  
For a point \( u \in M_f \), the \( \varepsilon \)-ball is defined as the connected component of the \( \varepsilon \)-slab \( f^{-1}([f(u) - \varepsilon, f(u) + \varepsilon]) \) that contains \( u \), denoted by \( B_\varepsilon(u, M_f) \).  
The parameter \( \tau \) represents the maximum of the total degrees of the nodes contained in \( B_\varepsilon(u, M_f) \), taken over all points \( u \) in \( M_f \) and \( M_g \).  
In contrast, our parameters \( \eta_f \) and \( \eta_g \) correspond to the numbers of leaf nodes of \( M_f \) and \( M_g \), respectively, and remain independent of \( \varepsilon \).
\begin{figure}[h]
    \centering
    \includegraphics[width=0.9\linewidth]{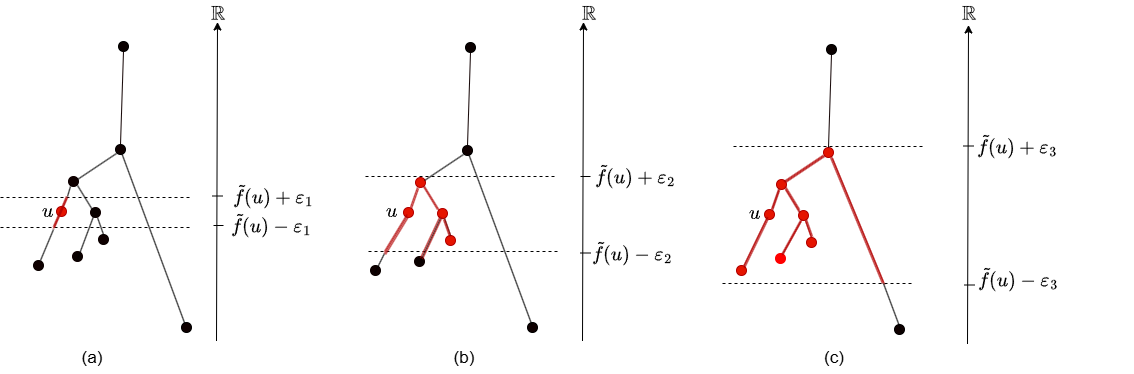}
    \caption{
        Illustration of the dependence of the parameter \( \tau \) for increasing values of $\varepsilon$ (\( \varepsilon_1 < \varepsilon_2 < \varepsilon_3 \)),  as shown in the corresponding figures (a), (b) and (c), respectively. Let \( u \) be a node of the given merge tree \( M_f \). 
        The red-colored portion represents  \( B_\varepsilon(u, M_f) \), while the black-colored portion lies outside  \( B_\varepsilon(u, M_f) \). 
        As \(\varepsilon \) increases (\( \varepsilon_1 < \varepsilon_2 < \varepsilon_3 \)), the number of nodes included in \( B_\varepsilon(u, M_f) \) grows as $2, 4, 7$. Consequently, the total degree of the nodes within \( B_\varepsilon(u, M_f) \) also increases as $2$, $9$ and $14$, respectively.
    }
    \label{fig:delta-growth}
\end{figure}

 As \( \varepsilon \) increases, the size of each \( \varepsilon \)-slab expands and thereby increasing the number of nodes contained in each \( \varepsilon \)-ball as illustrated in Figure~\ref{fig:delta-growth}. Consequently, the value of \( \tau \) increases.  In the extreme case where \( \varepsilon \) equals or exceeds the height of the tree, the \( \varepsilon \)-ball around any point \( u \in M_f \) includes all nodes in the tree.  
Hence, \( \tau \) equals the total degree of all nodes in \( M_f \), which is typically much larger than the numbers of leaf nodes \( \eta_f \) and \( \eta_g \).

 \begin{figure}[h]
     \centering
     \includegraphics[width=0.9\linewidth]{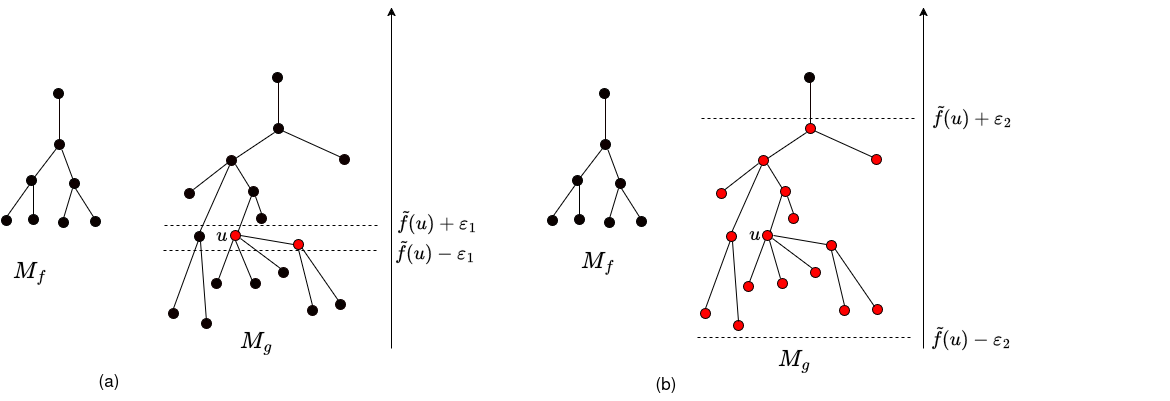}
     \caption{  Merge trees \( M_f \) and \( M_g \) with numbers of leaf nodes \( \eta_f = 4 \) and \( \eta_g = 10 \), respectively.  
        For a smaller value of \( \varepsilon = \varepsilon_1 \), we have \( \tau = 8 \) which is the total degree of the nodes colored in red in figure (a);  
        for a larger value \( \varepsilon = \varepsilon_2 \), \( \tau = 31 \) which is the total degree of the nodes colored in red in figure (b).}
     \label{fig:Eta-f-Eta-g}
 \end{figure}
 
 \color{black}
Consider the merge trees \( M_f \) and \( M_g \) shown in Figure~\ref{fig:Eta-f-Eta-g}.  
The total number of nodes across both trees is \( n = 25 \).  
For a smaller value \( \varepsilon = \varepsilon_1 \), we have \( \tau = 8 \).  
The time complexity of the algorithm in \cite{touli2022fpt} is given by
 \[
 O(2^{2\tau} (2\tau)^{2\tau+2} \cdot  n^2 \log^3 n) \approx 1.94\times10^{31}.
 \]
 For a larger value \( \varepsilon = \varepsilon_2 \), we have \( \tau = 31 \), yielding
  \[
 O(2^{2\tau} (2\tau)^{2\tau+2} \cdot  n^2 \log^3 n) \approx 1.49 \times 10^{138}.
 \]

In both cases, the numbers of leaf nodes in \( M_f \) and \( M_g \) remain unchanged (\( \eta_f = 4 \) and \( \eta_g = 10 \)) because these parameters are  independent of \( \varepsilon \).  
Since \( \eta_g^{\eta_f} \leq \eta_f^{\eta_g} \) (i.e., \( 10^4 \leq 4^{10} \)), we consider mappings from \( M_f \) to \( M_g \).  
For \( n = 25 \), \( \eta_f = 4 \), and \( \eta_g = 10 \), the complexity of our algorithm is
 \[
 O\!\left(n^2 \log n \;+\; \eta_g^{\eta_f}(\eta_f + \eta_g) \cdot n \log n\right) \approx 1.62 \times 10^7
 \]

which is significantly smaller than the complexity of the previous algorithm. In the refined version of our algorithm, the exponential term depends on \( \kappa^{\eta_f} \), where \( \kappa \) denotes the maximum size of the refined target node set for each leaf node in \( M_f \).  
For small values of \( \kappa \) (e.g., \( 2, 3, 4, 5 \)), this exponential dependence is substantially weaker than in the original term \( \eta_g^{\eta_f} \), resulting in a significant reduction in practical computational cost.  
Table~\ref{tab:kappa-values} summarizes the approximate total complexities of the refined algorithm, $O(n^2\log n + n\eta_f^2\eta_g^2 + \kappa^{\eta_f}(\eta_f + \eta_g) \cdot n\log n),$
for \( n = 25 \), \( \eta_f = 4 \), \( \eta_g = 10 \), and different values of \( \kappa \).

\begin{table}[h!]
\centering
\begin{tabular}{|c|c|}
\hline
$\kappa$ & Total Complexity \\
\hline
2 & $1.65 \times 10^{5}$ \\
3 & $5.43 \times 10^{5}$ \\
4 & $1.50 \times 10^{6}$ \\
5 & $3.67 \times 10^{6}$ \\
\hline
\end{tabular}

\caption{Approximate total complexity of the refined algorithm for various values of $\kappa$
\label{tab:kappa-values}}
\end{table}

This analysis clearly demonstrates that the parameter \( \tau \) grows rapidly with \( \varepsilon \), leading to prohibitively large computational costs even for moderate values of $\tau$.  
In contrast, our parameters \( \eta_f \) and \( \eta_g \) remain structurally bounded and stable across any values of \( \varepsilon \), which provides strong motivation for adopting the proposed parameterization in both theoretical and practical algorithmic settings.
 
 \color{black}
\section{Conclusion}
\label{sec:conclusion}
We presented a novel fixed-parameter tractable (FPT) algorithm for computing the exact interleaving distance between two merge trees. Our approach addresses the dependence of the degree-bound parameter $\tau$ on the choice of candidate values $\varepsilon$  and reparameterizes the problem in terms of the numbers of leaf nodes $\eta_f$ and $\eta_g$ of the merge trees $M_f$ and $M_g$, respectively,  which are independent of $\varepsilon$.  Building on the framework of $\varepsilon$-good maps, our method begins by defining a map on the leaf nodes of the first tree and then extending it along the leaf-to-root paths to construct a map between the merge trees. We then verify the existence of an $\varepsilon$-good map. This yields a substantial reduction in the polynomial and exponential component of the \textcolor{black}{complexity}, making our algorithm significantly more efficient. By parameterizing the problem using $\eta_f$ and $\eta_g$, we obtain a stable complexity bound that improves upon previous approaches.

Our theoretical contributions include a new characterization of well-definedness and continuity for the constructed map, grounded in the Gluing Lemma. We also provide efficient procedures to verify the properties required for a map to be $\varepsilon$-good. The complexity comparisons suggest that our method outperforms prior work, since $\eta_f$ and $\eta_g$ are independent of the choice of $\varepsilon$. Our work bridges the gap between theoretical soundness and computational practicality, enabling the computation of topologically meaningful distance between \textcolor{black}{two} merge trees. Future directions include extending this framework to Reeb graphs and adapting the method to analyze time-varying scalar fields.

\paragraph*{Acknowledgment.}
The first two authors would like to thank the MINRO Center (Machine Intelligence and Robotics Center) at  International Institute of Information Technology-Bangalore (IIITB), for funding this project. Furthermore, the third author has been supported in part by JSPS KAKENHI Grant Numbers JP22K18267, JP23H05437.

\bibliographystyle{abbrv}
\bibliography{references/Distance-Measures, references/morozov}
\end{document}